\documentclass[a4paper,12pt]{article}
\usepackage[
    DIV=11,
    headinclude=false,
    footinclude=false,
]{typearea}

\usepackage[T1]{fontenc}
\usepackage[utf8]{inputenc}
\usepackage[main=english]{babel}
\usepackage{csquotes}
\usepackage{microtype}

\usepackage{mathtools}
\mathtoolsset{centercolon,mathic}
\allowdisplaybreaks[2] 
\usepackage{amssymb}
\usepackage{amsthm}
\usepackage{thmtools}
\usepackage{derivative}
\usepackage{enumitem}
\usepackage{booktabs}
\usepackage{makecell}

\newcommand{\statement}[1]{\paragraph{#1}\pdfbookmark[1]{#1}{#1}} 

\renewcommand{\Delta}{\varDelta}

\renewcommand{\Phi}{\varPhi}

\renewcommand{\Psi}{\varPsi}

\renewcommand{\Lambda}{\varLambda}

\renewcommand{\Gamma}{\varGamma}

\renewcommand{\Omega}{\varOmega}

\let\epsi\varepsilon

\DeclarePairedDelimiter{\intervaloo}{\lparen}{\rparen}
\DeclarePairedDelimiter{\intervaloc}{\lparen}{\rbrack}

\DeclarePairedDelimiter{\intervalcc}{\lbrack}{\rbrack}

\DeclarePairedDelimiter{\abs}{\lvert}{\rvert}
\DeclarePairedDelimiter{\paren}{\lparen}{\rparen}

\DeclarePairedDelimiter{\norm}{\lVert}{\rVert}

    \newcommand{\VERT}[1]{#1|\mkern-1.5mu#1|\mkern-1.5mu#1|}

    \ExplSyntaxOn
    \makeatletter

    \MT_delim_default_inner_wrappers:n{\tnorm}

    \NewDocumentCommand{\tnorm}{ s o m }{
        \IfBooleanTF{#1}{
        \MT_delim_tnorm_star_wrapper:nnn%
            {\VERT{\bgroup\left}}{#3}{\VERT{\aftergroup\egroup\right}}
        }{
            \IfValueTF{#2}{
                \@nameuse{MT_delim_tnorm_nostarscaled_wrapper:nnn}%
                    {\VERT{\@nameuse {\MH_cs_to_str:N #2 l}}}
                    {#3}
                    {\VERT{\@nameuse {\MH_cs_to_str:N #2 r}}}
            }{
                \MT_delim_tnorm_nostarnonscaled_wrapper:nnn%
                    {\VERT{}}
                    {#3}
                    {\VERT{}}
            }
        }
    }

    \makeatother
    \ExplSyntaxOff

\DeclarePairedDelimiter{\commutator}{\lbrack}{\rbrack}

\DeclarePairedDelimiter{\List}{\{}{\}}

\DeclarePairedDelimiter{\braket}{\langle}{\rangle}
\DeclarePairedDelimiter{\bra}{\langle}{\rvert}
\DeclarePairedDelimiter{\ket}{\lvert}{\rangle}
\DeclareDocumentCommand{\ketbra}{o m m}{
    \IfValueTF{#1}
        {\ket[#1]{#2}\bra[#1]{#3}}
        {\ket{#2}\bra{#3}}%
}
\DeclareDocumentCommand{\proj}{o m}{
    \IfValueTF{#1}
        {\ket[#1]{#2}\bra[#1]{#2}}
        {\ket{#2}\bra{#2}}%
}

\DeclarePairedDelimiterXPP{\dd}[1]{d}{\lparen}{\rparen}{}{#1}
\DeclarePairedDelimiterXPP{\dist}[1]{d}{\lparen}{\rparen}{}{#1}
\DeclarePairedDelimiterXPP{\diam}[1]{\SYMdiam}{\lparen}{\rparen}{}{#1}

\makeatletter
\let\bBigg@@\bBigg@
\renewcommand{\bBigg@}[2]{{%
  \mathchoice
    {\bBigg@@{#1}{#2}}%
    {\bBigg@@{#1}{#2}}%
    {\big@size=.5\big@size\bBigg@@{#1}{#2}}%
    {\big@size=.3\big@size\bBigg@@{#1}{#2}}}}%
\makeatother

\DeclareDocumentCommand{\trace}{s o e{_} m}{
    \IfValueTF{#3}
        {\Tr_{#3}}
        {\Tr}%
    \IfBooleanTF{#1}
        {\paren*{#4}}
        {
            \IfValueTF{#2}
                {\paren[#2]{#4}}
                {\paren{#4}}%
        }%
}
\DeclareDocumentCommand{\Trace}{s o e{_} m}{
    \IfValueTF{#3}
        {\Tr_{#3}}
        {\Tr}%
    \IfBooleanTF{#1}
        {\paren*{#4}}
        {
            \IfValueTF{#2}
                {\paren[#2]{#4}}
                {\paren{#4}}%
        }%
}

\providecommand\given{}

\newcommand\SetSymbol[1][]{%
    \nonscript\,#1\vert
    \allowbreak
    \nonscript\,
    \mathopen{}}
\DeclarePairedDelimiterX\Set[1]\{\}{%
    \renewcommand\given{%
        \SetSymbol[\delimsize]}
    \nonscript\,
    #1
    \nonscript\,
}

\ExplSyntaxOn
\makeatletter
\cs_generate_variant:Nn \tl_if_eq:nnT {onT}
\tl_new:N \__scale_new_delimsize_tl
\cs_new:Nn \__scale_make_bigger_tl:N {
    \tl_if_eq:onT {#1} {\@MHempty} { \tl_set:Nn \__scale_new_delimsize_tl { big } }
    \tl_if_eq:onT {#1} {}          { \tl_set:Nn \__scale_new_delimsize_tl { big } }
    \tl_if_eq:onT {#1} {\big}      { \tl_set:Nn \__scale_new_delimsize_tl { Big } }
    \tl_if_eq:onT {#1} {\Big}      { \tl_set:Nn \__scale_new_delimsize_tl { bigg } }
    \tl_if_eq:onT {#1} {\bigg}     { \tl_set:Nn \__scale_new_delimsize_tl { Bigg } }
    \tl_if_eq:onT {#1} {\middle}   { \tl_set:Nn \__scale_new_delimsize_tl { middle } }
}
\cs_new:Nn \scale_make_bigger_m:N {
    \__scale_make_bigger_tl:N #1
    \use:c { \tl_use:N \__scale_new_delimsize_tl }
}
\cs_new:Nn \scale_make_bigger_l:N {
    \__scale_make_bigger_tl:N #1
    \tl_if_eq:NnTF \__scale_new_delimsize_tl {middle} {
        \tl_set:Nn \__scale_new_delimsize_tl {left}
    }{
        \tl_put_right:Nn \__scale_new_delimsize_tl {l}
    }
    \use:c { \tl_use:N \__scale_new_delimsize_tl }
}
\cs_new:Nn \scale_make_bigger_r:N {
    \__scale_make_bigger_tl:N #1
    \tl_if_eq:NnTF \__scale_new_delimsize_tl {middle} {
        \tl_set:Nn \__scale_new_delimsize_tl {right}
    }{
        \tl_put_right:Nn \__scale_new_delimsize_tl {r}
    }
    \use:c { \tl_use:N \__scale_new_delimsize_tl }
}
\DeclarePairedDelimiterXPP{\pd}[1]{\scale_make_bigger_l:N\delimsize\lparen d}{\lparen}{\rparen}{\scale_make_bigger_r:N\delimsize\rparen}{#1}
\DeclarePairedDelimiterXPP{\pdist}[1]{\scale_make_bigger_l:N\delimsize\lparen d}{\lparen}{\rparen}{\scale_make_bigger_r:N\delimsize\rparen}{#1}
\DeclarePairedDelimiterXPP{\pdiam}[1]{\scale_make_bigger_l:N\delimsize\lparen \SYMdiam}{\lparen}{\rparen}{\scale_make_bigger_r:N\delimsize\rparen}{#1}

\makeatother
\ExplSyntaxOff

\usepackage[svgnames]{xcolor}

\usepackage[
    colorlinks,
    allcolors=col_1,
    bookmarksnumbered,
    bookmarksopen,
    bookmarksopenlevel=2,
    bookmarksdepth=3,
    pdfusetitle=true,
    breaklinks=true,
]{hyperref}

\makeatletter
\@ifpackageloaded{hyperref}{
    \usepackage[capitalise,noabbrev]{cleveref}
    \usepackage{orcidlink}
}{
    \usepackage[capitalise,noabbrev,poorman]{cleveref}
    \newcommand{\texorpdfstring}[2]{#1}
    \newcommand{\href}[2]{#2}
    \newcommand{\hypersetup}[1]{}
    \newcommand{\orcidlink}[1]{ORCiD}
    \newcommand{\pdfbookmark}[1]{}
}
\makeatother

\crefdefaultlabelformat{#2#1#3}
\crefname{equation}{}{}
\renewcommand{\namecref}{\lcnamecref}

\usepackage{caption}
\usepackage{subcaption}

\makeatletter
\newcommand*{\diff}{\@ifnextchar^{\DIfF}{\DIfF^{}}}
\def\DIfF^#1{%
    \mathop{\mathrm{\mathstrut d}}\nolimits^{#1}\gobblespace}
\def\gobblespace{\futurelet\diffarg\opspace}
\def\opspace{%
    \let\DiffSpace\!%
    \ifx\diffarg(%
        \let\DiffSpace\relax
    \else
        \ifx\diffarg[%
            \let\DiffSpace\relax
        \else
            \ifx\diffarg\{%
                \let\DiffSpace\relax
            \fi
        \fi
    \fi
    \DiffSpace
} 
\makeatother

\newcommand{\sumstack}[2][]{\ifstrempty{#1}{\sum_{\substack{#2}}}{\smashoperator[#1]{\sum_{\substack{#2}}}}}

\newcommand{\e}{{\mathrm{e}}}
\newcommand{\I}{\mathrm{i}}
 \newcommand{\R}{\mathbb{R}}
\newcommand{\C}{\mathbb{C}}
\newcommand{\N}{\mathbb{N}}
\newcommand{\Z}{\mathbb{Z}}
\newcommand{\Hi}{{\mathcal{H}}}
\newcommand{\alg}{\mathcal{A}}
\newcommand{\algloc}{\alg^{\mathup{loc}}}
\newcommand{\unit}{\mathbf{1}}
\DeclareMathOperator{\id}{id}
\newcommand\cexpsym{\mathbb{E}}
\ExplSyntaxOn
\DeclareDocumentCommand{\cexp}{s o m m}{%
    \cexpsym\c_math_subscript_token{#3}
    \IfBlankF{#4}
    {
        \exp_last_unbraced:Ne \paren {\IfBooleanT{#1}{*}\IfValueT{#2}{[\exp_not:N #2]}} {#4}
    }
}
\ExplSyntaxOff

\newcommand{\calF}{\mathcal{F}}

\DeclareMathOperator{\Tr}{Tr}
\newcommand{\SYMdiam}{\operatorname{diam}}

\DeclareMathOperator{\supp}{supp}

\newcommand{\supint}{^{\mathup{int}}}
\newcommand{\Phiint}{\Phi\supint}
\newcommand{\tauint}{\tau\supint}
\newcommand{\Hint}{H\supint}

\newcommand{\ketU}{\ket{\Uparrow}}
\newcommand{\ketD}{\ket{\Downarrow}}
\newcommand{\braU}{\bra{\Uparrow}}
\newcommand{\braD}{\bra{\Downarrow}}
\newcommand{\ua}{\uparrow}
\newcommand{\da}{\downarrow}

\newcommand{\pauli}[1]{\sigma^{#1}}
\newcommand{\componentX}{X}

\newcommand{\componentZ}{Z}
\newcommand{\pauliX}{\pauli{\componentX}}

\newcommand{\pauliZ}{\pauli{\componentZ}}

\newcommand{\suchthat}{\mathpunct{\ordinarycolon}}

\newcommand{\quadtext}[1]{\quad\text{#1}\quad}
\newcommand{\qquadtext}[1]{\quad\quadtext{#1}\quad}

\newcommand{\Alignindent}{\hspace*{2em}&\hspace*{-2em}}
\newcommand{\mathup}[1]{\mathrm{#1}}

\let\oldsetminus\setminus
\newbox\mybox
\newcommand\cutsetminus[1]{%
    \setbox\mybox\hbox{\(#1\oldsetminus\)}%
    \ht\mybox=0pt%
    \dp\mybox=0pt%
    \usebox\mybox%
}
\renewcommand\setminus{%
    \mathbin{%
        \mathchoice%
            {\displaystyle\oldsetminus}
            {\textstyle\oldsetminus}
            {\cutsetminus{\scriptstyle}}
            {\cutsetminus{\scriptscriptstyle}}
    }%
}

\makeatletter
\ExplSyntaxOn

\newcounter{theoremenv}
\newcounter{theoremenvglobal}
\newlist{thmlist}{enumerate}{1}
\setlist[thmlist]{
    label=\textup{(\alph{thmlisti})},
    ref={(\alph{thmlisti})},
    nosep,
}
\renewcommand{\p@thmlisti}{\perh@ps{\protect\ref{auto-label:\arabic{theoremenvglobal}}}}
\DeclareRobustCommand{\perh@ps}[1]{#1}
\newcommand{\itemref}[1]{%
    \begingroup 
    \let\perh@ps\@gobble\ref{#1}%
    \endgroup
}

\addtotheorempostheadhook{
    \crefalias{thmlisti}{\thmt@envname}
    \setcounter{theoremenv}{\value{\thmt@envname}}
    \stepcounter{theoremenvglobal}
    \exp_args:NNe \renewcommand \thetheoremenv {\use:c{the\thmt@envname}}
    \exp_args:Ne \label {auto-label:\arabic{theoremenvglobal}}
}
\makeatother
\ExplSyntaxOff

\theoremstyle{plain}

\declaretheorem[
    name=Theorem,
]{theorem}

\declaretheorem[
    name=Lemma,
    sibling=theorem,
]{lemma}

\declaretheorem[
    name=Corollary,
    sibling=theorem,
]{corollary}

\theoremstyle{definition}

\declaretheorem[
    name=Definition,
    sibling=theorem,
]{definition}

\theoremstyle{remark}
\declaretheorem[
    name=Remark,
    sibling=theorem,
    qed=\(\diamond\),
]{remark}

\declaretheorem[
    name=Example,
    sibling=theorem,
    qed=\(\diamond\),
]{example}

\usepackage{tikz}
\usetikzlibrary{
    arrows.meta,
    calc,
    decorations.pathreplacing,
    shadings,
}
\pgfdeclarelayer{background}
\pgfsetlayers{background,main}
\newlength{\ul}
\newlength{\lw}
\tikzset{%
    myLine/.style = {
        line width = #1\lw,
    },
    myLine/.default = 1,
    x = \ul,
    y = \ul,
    text = black,
    draw = black,
    shorten/.style = {
        shorten < = #1,
        shorten > = #1
    },
    shorten/.default = 2\lw,
    on layer/.code = {
        \pgfonlayer{#1}\begingroup
        \aftergroup\endpgfonlayer
        \aftergroup\endgroup
    },
    on background/.style = {
        preaction = {
            #1,
            on layer = background,
        },
    },
    only inside/.style = {
        preaction = {
            clip,
            postaction = {
                line width = 2\lw,
                #1
            },
        },
    },
    only inside/.default = {draw},
    >={Straight Barb},
}

\definecolorset{HTML}{col_}{}{%
    1,0072B2;%
    2,D55E00;%
    3,E69F00
}

\definecolor{col_nobound}{rgb:hsb}{.03,.5,.85}
\definecolor{col_bound}{rgb:hsb}{.16,.6,1}
\definecolor{col_linear}{rgb:hsb}{.4,1,.85}
\definecolor{col_algebraic}{rgb:named}{col_bound!60!col_linear}

\tikzset{%
    every node/.style = {
        node font = \sffamily,
        inner sep = 4\lw,
    },
    small/.style = {
        font=\small,
    },
}

\usepackage[
    bibstyle=phys,
    citestyle=numeric,
    biblabel=brackets,
    mincitenames=1,
    maxcitenames=3,
    minalphanames=3,
    maxalphanames=4,
    minbibnames=3,
    maxbibnames=10,
    arxiv=abs,
    eprint=true,
    doi=false,
    url=false,
    sorting=nty,
    useprefix=true,
    alldates=year,
    urldate=long,
]{biblatex}
\renewbibmacro*{note+pages}{%
    \printfield{note}%
    \setunit{\bibpagespunct}%
    \printfield{pages}%
    \newunit%
}
\makeatletter
\DeclareFieldFormat{eprint:arxiv}{%
    \ifhyperref
    {\href{https://arxiv.org/\abx@arxivpath/#1}{%
            arXiv\addcolon\linebreak[2]#1}}
    {arXiv\addcolon
        \nolinkurl{#1}%
    }
}
\makeatother

\renewbibmacro*{maintitle+title}{%
    \iffieldsequal{maintitle}{title}{%
        \clearfield{maintitle}%
        \clearfield{mainsubtitle}%
        \clearfield{maintitleaddon}%
    }{%
        \iffieldundef{maintitle}{}{%
            \printtext[doi/url-link]{%
                \usebibmacro{maintitle}%
            }%
            \newunit\newblock
            \iffieldundef{volume}{}{%
                \printfield{volume}%
                \printfield{part}%
                \setunit{\addcolon\space}%
            }%
        }%
    }%
    \printtext[doi/url-link]{%
        \usebibmacro{title}%
    }%
    \newunit%
}

\DeclareBibliographyDriver{article}{%
    \usebibmacro{bibindex}%
    \usebibmacro{begentry}%
    \usebibmacro{author/translator+others}%
    \setunit{\labelnamepunct}\newblock
    \usebibmacro{title}%
    \newunit
    \printlist{language}%
    \newunit\newblock
    \usebibmacro{byauthor}%
    \newunit\newblock
    \usebibmacro{bytranslator+others}%
    \newunit\newblock
    \printfield{version}%
    \newunit\newblock
    \printtext[doi/url-link]{%
        \usebibmacro{journal+issuetitle}%
        \newunit
        \usebibmacro{byeditor+others}%
        \newunit
        \usebibmacro{note+pages}%
        \newunit\newblock
        \iftoggle{bbx:isbn}
        {\printfield{issn}}
        {}%
    }%
    \newunit\newblock
    \printfield{pubstate}%
    \setunit{\addspace}%
    \printfield{year}%
    \newunit\newblock
    \usebibmacro{doi+eprint+url}%
    \newunit\newblock
    \printfield{addendum}%
    \setunit{\bibpagerefpunct}\newblock
    \usebibmacro{pageref}%
    \newunit\newblock
    \iftoggle{bbx:related}
        {\usebibmacro{related:init}%
        \usebibmacro{related}}
        {}%
    \usebibmacro{finentry}%
}

\usepackage{xurl} 

\usepackage{dsfont}

\title{Enhanced Lieb-Robinson bounds for\texorpdfstring{\\}{ }commuting long-range interactions}

\usepackage{authblk}

\author{Marius Lemm%
    \texorpdfstring{%
        \,\orcidlink{0000-0001-6459-8046}
        \thanks{\href{mailto:marius.lemm@uni-tuebingen.de}{marius.lemm@uni-tuebingen.de}}%
    }{}%
}

\author{Tom Wessel%
    \texorpdfstring{%
        \,\orcidlink{0000-0001-7593-0913}
        \thanks{\href{mailto:tom.wessel@uni-tuebingen.de}{tom.wessel@uni-tuebingen.de}}%
    }{}%
}

\affil{
Department of Mathematics, University of Tübingen, 72076 Tübingen, Germany
}

\date{September 2025}

\begin{document}

\bgroup
\hypersetup{hidelinks}
\maketitle\thispagestyle{empty}
\egroup

\begin{abstract}
    Recent works have revealed the intricate effect of long-range interactions on information transport in quantum many-body systems:
    In \(D\) spatial dimensions, interactions decaying as a power-law \(r^{-\alpha}\) with \(\alpha > 2 D+1\) exhibit a Lieb-Robinson bound (LRB) with a linear light cone and the threshold \(2D +1\) is sharp in general.
    Here, we observe that mutually commuting, long-range interactions satisfy an enhanced LRB of the form \(t \, r^{-\alpha}\) for any \(\alpha>0\), and this scaling is sharp.
    In particular, the linear light cone occurs at \(\alpha = 1\) in any dimension.
    Part of our motivation stems from quantum error-correcting codes.
    As applications, we derive enhanced bounds on ground state correlations and an enhanced local perturbations perturb locally (LPPL) principle for which we adapt a recent subharmonicity argument of Wang-Hazzard.
    Similar enhancements hold for commuting interactions with stretched exponential decay.
\end{abstract}

\section{Introduction}
\label{sec:introduction}

The Lieb-Robinson bound (LRB) is a central theorem in quantum many-body physics.
It underpins the proofs of fundamental structural properties such as the decay of correlations of ground states of Hamiltonians with a spectral gap~\cite{HK2006,NS2006}, the existence of the thermodynamic limit as a strongly continuous automorphism on the quasi-local operator algebra~\cite{BR1981,NOS2006}, and the mathematical definition of a topological quantum phase~\cite{HW2005,BMNS2012}.
At the same time, the LRB is a cornerstone of the growing area of quantum information theory: on the one hand, it was a decisive tool for landmark results such as the area law for the entanglement entropy~\cite{Hastings2007area} and on the other hand it provides bounds on dynamical entanglement generation~\cite{BHV2006}, quantum messaging and quantum state transfer~\cite{EW2017}, and efficient quantum simulability of many-body dynamics~\cite{HHKL2021}.

The LRB is concerned with quantum lattice Hamiltonians.
These are defined by fixing a finite graph \(\Lambda\) and defining the Hilbert space
\begin{equation*}
    \Hi_\Lambda := \bigotimes_{x\in \Lambda} \C^q
\end{equation*}
with the standard inner product.
On this Hilbert space, one considers a distinguished self-adjoint linear operator, the Hamiltonian \(H\), which is taken as a sum of terms that act \enquote{locally}, i.e.
\begin{equation}
    \label{eq:Hintro}
    H_\Lambda := \sum_{Z\subset \Lambda} \Psi(Z)
    ,
\end{equation}
where each summand is a bounded, linear, self-adjoint operator of the form \(\Psi(Z) \equiv \Psi(Z) \otimes \unit_{\Lambda\setminus Z}\).
The precise mathematical setup in terms of quasi-local algebras is given in \cref{sec:setup}.

Instead of studying solutions to the Schrödinger equation directly, it is advantageous to study its dual, the quantum dynamics on operators, known as \emph{Heisenberg dynamics}, that is defined by unitary conjugation with the solution operator of the Schrödinger equation, i.e.
\begin{equation}
    \label{eq:timeintro}
    \tau_t^{\Lambda}(A)
    =
    \e^{\I t H_{\Lambda}} \, A \, \e^{-\I t H_{\Lambda}}
    .
\end{equation}
For a set \(X \subset \Lambda\) and \(\ell>0\), we denote its \(\ell\)-neighbourhood by
\begin{equation}
    \label{eq:def-neighbourhood}
    X_\ell := \Set[\big]{x \in \Lambda \given \dist{x,X} \leq \ell}
    .
\end{equation}

To put our results in context, we first review the two kinds of LRBs that are well-established for the case where the interaction \(\Psi\) is of finite-range, i.e.\ there exists a range \(R>0\) such that \(\Psi(Z)=0\) if \(\diam{Z}>R\), or exponentially decaying, i.e.\ \(\norm{\Psi(Z)}\) decays exponentially in~\(\diam{Z}\).
For such rapidly decaying interactions, one has the following two types of LRBs~\cite{LR1972,HK2006,NS2006,NSY2019}.

\begin{enumerate}[label={(\alph*)}]
    \item \textbf{Commutator version of the LRB.}
        There exist constants \(C\), \(b\), \(v>0\) such that for all finite \(X\), \(Y\subset \Lambda\) and bounded operators \(A_X \equiv A_X \otimes \unit_{\Lambda\setminus X}\) and \(B_Y \equiv B_Y \otimes \unit_{\Lambda\setminus Y}\), it holds that
        \begin{equation}
            \label{eq:LRBcomm}
            \norm[\big]{
                \commutator{\tau_t^\Lambda(A_X),B_Y}
            }
            \leq
            C \, \min\List[\big]{\abs{X},\abs{Y}} \, \norm{A_X} \, \norm{B_Y} \, \e^{b\, (vt-\dist{X,Y})}
            .
        \end{equation}
    \item \textbf{Operator localization version of the LRB.}
        There exist constants \(C\), \(b\), \(v>0\) such that for all finite \(X\subset \Lambda\) and bounded operators \(A_X \equiv A_X \otimes \unit_{\Lambda\setminus X}\), it holds that
        \begin{equation}
            \label{eq:LRBloc}
            \norm[\big]{
                \tau^\Lambda_t(A_X)
                -\cexp[\big]{X_r}{\tau_t^{\Lambda}(A_X)}
            }
            \leq
            C \, \abs{X} \, \norm{A_X} \, \e^{b(vt-r)}
            .
        \end{equation}
        Here, \(\cexp{X_r}{}\) is the conditional expectation on \(X_r\), i.e.\ the partial trace over the Hilbert space associated to \(\Lambda\setminus X_r\); see~\eqref{eq:definition-conditional-expectation}.
\end{enumerate}

The LRBs are a mathematically precise way to capture that quantum information propagates at most with a speed \(v\) in these systems.
The speed~\(v\) and the constants \(C\) and \(b\) might depend on the details of the interaction but are uniformly bounded in the system size \(\abs{\Lambda}\).
Indeed, since the right-hand sides are small for \(v \, t \ll \dist{X,Y}\) and \(v \, t \ll r\), respectively, the LRBs establish the existence of an effective finite propagation speed for many-body Hamiltonians of the form~\eqref{eq:Hintro}.
The fact that the bound is effective until \(v t \sim r\) is also called a ``linear light cone'' (with slope \(v\)) or a ``ballistic bound'' and we use both of these phrases interchangeably.

\subsection{Enhanced Lieb-Robinson bounds}

Given the important role played by LRBs, it is unsurprising that significant effort has been invested to extend its validity to other classes of many-body systems.
Of particular interest in the past ten~years has been the extension of LRBs to so-called \emph{long-range} interactions, whose operator norm decays as a power law.
In this introduction, we focus for simplicity on two-body interactions, i.e.\ \(\Psi(Z)=0\) for \(\abs{Z}>2\), and we refer to \cref{thm:LR-bound-long-range-alpha-ge-0,thm:LRB-long-range-alpha-ge-nu} for the more general statements.
For two-body interactions, the power-law decay assumption can be simply phrased as the operator norm bound
\begin{equation}
    \label{eq:powerlaw}
    \norm{\Psi(\List{x,y})}
    \leq
    C \, \dist{x,y}^{-\alpha}
\end{equation}
and \(\norm{\Psi(\List{x})} \leq C\).

After intensive research efforts and the introduction of a variety of new techniques, the sharp form of the LRBs for long-range interactions~\eqref{eq:powerlaw} were nailed down~\cite{EWM2013,RGLS2014,EGYM2017,EMN2020,KS2020,TCE2020,TGDL2021,TGB2021,RS2024}.
Also, propagation bounds for Bose-Hubbard type Hamiltonians with unbounded long-range hopping terms have recently been studied~\cite{FLS2022lieb,FLS2022maximal,LRSZ2023,LRZ2023,VKS2024,LRZ2025}.

Our goal is to study the effect of the additional assumption that the interactions mutually commute, i.e.
\begin{equation}
    \label{eq:comm}
    \commutator[\big]{\Psi(X),\Psi(Y)}=0
    \qquad\text{for all \(X\), \(Y\subset \Lambda\)}
    .
\end{equation}
The commutativity~\eqref{eq:comm} is a significant restriction, but it is physically very relevant in the context of quantum error correcting codes as we explain in Section~\ref{sec:motivation}.
Commutativity~\eqref{eq:comm} allows us to prove significantly enhanced LRBs compared to general long-range interactions, which satisfy only~\eqref{eq:powerlaw} but not necessarily~\eqref{eq:comm}.
We now summarize the differences between previous bounds; see also Figure~\ref{fig:intro}.
Detailed statements are given in \cref{thm:LR-bound-long-range-alpha-ge-0,thm:LRB-long-range-alpha-ge-nu}.
The parameter \(D\) captures the dimension of the graph; think of \(\Z^D\) and see \cref{defn:graph}.

\newlength{\verticalDist}
\setlength{\verticalDist}{1.2cm}
\begin{figure}
    \begin{subcaptiongroup}
    \begin{center}
        \setlength{\ul}{.4cm}
        \begin{tikzpicture}[
                myLine,
                x=\ul,y=\ul,
                ticknotes/.style = {
                    every node/.append style = {
                        rectangle,
                        minimum height=1.5ex,
                        inner sep=0pt,
                        draw=black,
                        minimum width=0pt,
                        outer sep=\lw,
                    },
                    every label/.append style = {
                        label position=below,
                        inner sep=\lw,
                    },
                },
            ]
            \coordinate (LC2 scale) at (0,0);

            \draw[->,ticknotes]  (LC2 scale) node[label=\(0\)]     (0) {}
                -- +(6,0) coordinate       (1) {}
                -- ++(15,0) node[label=\(D\)]    (D) {}
                -- +(6,0) node[label=\(D+1\)]      (D+1) {}
                -- ++(15,0) node[label=\(2D\)]   (2D) {}
                -- ++(6,0) node[label=\(2D+1\)] (2D+1) {}
                -- ++(6,0) coordinate[label=-45:\(\alpha\)] (larger);

            \coordinate (1b) at ($ (LC2 scale) + (0,2ex) $);
            \coordinate (1t) at ($ (1b) + (0,\verticalDist)$);
            \node[above right] (label LC2 commuting) at (1t) {\dots{} commuting Hamiltonians:};

            \coordinate[above] (2b) at (label LC2 commuting.north west);
            \coordinate (2t) at ($ (2b) + (0,\verticalDist)$);
            \node[above right] (label LC2 general) at (2t) {\dots{} general Hamiltonians:};
            \node[above right] (caption LC2) at (label LC2 general.north west) {\llap{(b) }\phantomcaption Bound on approximation \(
                \norm[\big]{
                    \tau^\Lambda_t(A_X)
                    -\cexp[\big]{X_r}{\tau_t^{\Lambda}(A_X)}
                }
            \) for \dots{} \label{fig:intro-commutator-version}};

            \path[fill=col_nobound] (0.east |- 2b) rectangle node[small] {no bound} (D.west |- 2t);
            \path[fill=col_bound] (D.east |- 2b) rectangle node[small] {\(\displaystyle\frac{\e^{v t}}{r^{\alpha-D}}\)} (2D.west |- 2t);
            \path[fill=col_algebraic] (2D.east |- 2b) rectangle node[small] {\(\displaystyle\paren[\bigg]{\frac{t}{\!r^{\alpha-2D}}\!}^{\!\frac{\alpha-D}{\alpha-2D}}\)} (2D+1.west |- 2t);
            \path[fill=col_linear] (2D+1.east |- 2b) rectangle node[small] {\(\displaystyle\frac{t^{D+1}}{(r-vt)^{\alpha-D}}\)} (larger |- 2t);

            \path[fill=col_nobound] (0.east |- 1b) rectangle node[small] {no bound} (D.west |- 1t);
            \path[fill=col_linear] (D.east |- 1b) rectangle node {\(\displaystyle\frac{t}{r^{\alpha-D}}\)} (larger |- 1t);
            \path[fill=col_algebraic] (D.east |- 1b) rectangle (D+1 |- 1t);

            \coordinate[above=2\baselineskip] (LC1 scale) at (caption LC2.north west);

            \draw[->,ticknotes]  (LC1 scale) node[label=\(0\)] {}
                -- (1 |- LC1 scale) node[label=\(1\)]         {}
                -- (D |- LC1 scale) node[label=\(D\)]     {}
                -- (D+1 |- LC1 scale)
                -- (2D |- LC1 scale) node[label=\(2D\)]    {}
                -- (2D+1 |- LC1 scale) node[label=\(2D+1\)]  {}
                -- (larger |- LC1 scale) coordinate[label=-45:\(\alpha\)];

            \coordinate (3b) at ($ (LC1 scale) + (0,2ex) $);
            \coordinate (3t) at ($ (3b) + (0,\verticalDist)$);
            \node[above right] (label LC1 commuting) at (3t) {\dots{} commuting Hamiltonians:};
            \coordinate[above] (4b) at (label LC1 commuting.north west);
            \coordinate (4t) at ($ (4b) + (0,\verticalDist)$);
            \node[above right] (label LC1 general) at (4t) {\dots{} general Hamiltonians:};
            \node[above right] at (label LC1 general.north west) {\llap{(a) }\phantomcaption Bound on commutator \(\norm[\big]{\commutator{\tau^\Lambda_t(A),B}}\) for \dots{}\label{fig:intro-operator-localization-version}};

            \path[fill=col_nobound] (0.east |- 4b) rectangle node[small] {no bound} (D.west |- 4t);
            \path[fill=col_bound] (D.east |- 4b) rectangle node[small] {\(\displaystyle\frac{\e^{v t}}{r^{\alpha}}\)} (2D.west |- 4t);
            \path[fill=col_algebraic] (2D.east |- 4b) rectangle node[small] {\(\displaystyle\paren[\bigg]{\frac{t}{\!r^{\alpha-2D}}\!}^{\!\frac{\alpha-D}{\alpha-2D}}\)} (2D+1.west |- 4t);
            \path[fill=col_linear] (2D+1.east |- 4b) rectangle node[small] {\(\displaystyle\frac{t^{2D+1}}{(r-vt)^{\alpha}}\)} (larger |- 4t);

            \path[fill=col_linear] (0.east |- 3b) rectangle node {\(\displaystyle\frac{t}{r^\alpha}\)} (larger |- 3t);
            \path[fill=col_algebraic] (0.east |- 3b) rectangle (1 |- 3t);

            \coordinate (cbbnorth) at (current bounding box.north);
            \coordinate (cbbsouth) at (current bounding box.south);
            \pgfresetboundingbox
            \path[use as bounding box] (cbbnorth) rectangle (cbbsouth);
        \end{tikzpicture}
    \end{center}
    \end{subcaptiongroup}
    \caption{
        Comparison of the previously shown (sharp) Lieb-Robinson bounds for general long-range interactions and the enhanced Lieb-Robinson bounds proven here for long-range \emph{commuting} interactions.
        The general bounds mentioned above are proven in~\cite{TGB2021} for \(\alpha\in \intervaloo{2D,2D+1}\) and~\cite{KS2020} for \(\alpha>2D+1\).
        The bounds displayed for \(\alpha\in\intervaloo{D,2D}\) hold for all~\(\alpha>D\), the commutator version was proven in~\cite{HK2006} and the operator localization version in~\cite[section S.III.A]{KS2021}.
        The enhanced bounds for commuting Hamiltonians are stated in \cref{thm:LRB-long-range-alpha-ge-nu,thm:LR-bound-long-range-alpha-ge-0}.
    }
    \label{fig:intro}
\end{figure}

In \cref{fig:intro}, we compare the results for general long-range interactions~\cite{HK2006,NS2006,EWM2013,RGLS2014,EGYM2017,EMN2020,KS2020,TCE2020,TGDL2021,TGB2021,RS2024} to the ones we prove here for mutually commuting long-range interactions.
The LRBs for commuting long-range interactions are enhanced compared to the LRB for general long-range interactions in two ways.

\begin{enumerate}[label=(\roman*)]
    \item
        \emph{The range of \(\alpha\) where an LRB exists at all is increased.}
        For general interactions, there are no LRBs known at all for \(\alpha\leq D\).
        By contrast, for commuting long-range interactions, the commutator version of the LRB holds for any \(\alpha>0\) and the localization version holds for any \(\alpha>D\).
    \item \emph{The LRBs are stronger, meaning they are effective in larger spacetime regions.}
        The central first question is whether the light cone is indeed linear, i.e.\ of the form \(vt\sim r\).
        A LRB corresponding to a \emph{linear light cone} says that information transport is at most \emph{ballistic}, i.e.\ it can spread at most a distance proportional to \(t\) in time \(t\).
        We use the term linear light cone and ballistic LRB interchangeably, as both are common in the literature.
        In the finite-range case, both LRBs~\eqref{eq:LRBcomm} and~\eqref{eq:LRBloc} are ballistic because they become effective at distance \(vt\), which is linear in \(t\).
        For general long-range interactions, both versions of the LRB become ballistic for \(\alpha>2D+1\).
        For commuting interactions, the commutator version of the LRB is ballistic already at \(\alpha=1\) and the localization version is ballistic at \(\alpha = D+1\).
        Moreover, the LRBs for commuting interactions even become subballistic, i.e.\ the critical distance scales sublinearly in time, above these thresholds.
\end{enumerate}

The LRBs stated for general (non-commuting) long-range interactions in \cref{fig:intro-commutator-version} are proven to be sharp in the following sense: bounds which assume better scaling of \(t\) versus \(r\) can be violated by explicit choices of time-dependent Hamiltonians; see for example~\cite{EMN2020,TCE2020,TGDL2021}.
Similarly, the bounds for commuting interactions are also sharp, as we discuss in \cref{sec:optimality-of-the-LRB}.

The proofs we give are quite short and more direct than the usual proofs of LRBs, thanks to the commutativity.
We first prove a bound similar to the operator localization version of the LRB, \cref{thm:local-approximation-of-time-evolution}, which bounds \(\norm[\big]{\tau^\Lambda_t(A_X) - \tau^{\Lambda'}_t(A_X)}\) where the second evolution is with the Hamiltonian restricted to \(\Lambda'=\Lambda\setminus Y\) or \(\Lambda'=X_r\) to obtain the commutator or operator localization version, respectively.
This is different to the usual approach for non-commuting interactions, where one commonly first proves the commutator version and then derives the operator localization version.
In both approaches, one usually looses a power \(D\) in the decay for the operator localization compared to the commutator version.
Only for \(\alpha\in \intervaloo{2D,2D+1}\), the proof of the general Lieb-Robinson bound, see \cref{fig:intro}, directly works with the operator localization, which always implies the commutator version with the same decay.
Thus, the spatial decay in this regime is~\(r^{\alpha-D}\) for both LRBs and might be improved for the commutator version.

We remark in passing that a different direction in which the general bounds on long-range interactions can be improved is for non-interacting particles, which display a linear light cone at \(\alpha>D+1\)~\cite{TCE2020}.

\subsection{Further results}
\label{sec:intro-further-results}

As mentioned in the introduction, LRBs and operator localization bounds have a plethora of applications -- both structural ones in mathematical physics and practical ones in quantum information theory, e.g.\ on the speed of entanglement generation~\cite{BHV2006}, quantum messaging and quantum state transfer~\cite{EW2017}.
Here, we focus on two applications which use LRBs to constrain the structure of ground states of Hamiltonians with a uniform spectral gap, namely decay of correlations and the principle that ``local perturbations perturb locally'' (LPPL).

\subsubsection{Enhanced correlation bounds for gapped ground states}

Modifying the proof of \textcite{HK2006}, we show that a ground state~\(\rho_0\), which is gapped from the rest of the spectrum by a gap~\(g\), satisfies \emph{decay of correlations} in the sense that
\begin{equation*}
    \abs[\Big]{
        \trace[\big]{
            \rho_0
            \, A_X
            \, B_Y
        }
        - \trace[\big]{
            \rho_0
            \, A_X
        }
        \, \trace[\big]{
            \rho_0
            \, B_Y
        }
    }
    \lesssim
    \, \norm{A_X} \, \norm{B_Y}
    \, \abs{X} \, \abs{Y}
    \, g^{-1}
    \, \dist{X,Y}^{-\alpha}
\end{equation*}
for small \(g\) and any \(\alpha>0\).
The precise statement is given in \cref{thm:spectral-gap-implies-decay-of-correlations}.
In this bound, the inverse spectral gap \(g^{-1}\) plays a similar role as time in \cref{fig:intro}.
In particular, the analogue of a ballistic bound is that the relevant localization length scales like~\(g^{-1}\), which is the case for \(\alpha \geq 1\).
This is to be compared with the best known bound \(r^{-\alpha g/v}\) for general long-range interactions with \(\alpha\in\intervaloc{D,2D}\)~\cite{WH2022}.
First, no bound for \(\alpha \leq D\) is known in the general case.
And second, writing \(r^{-\alpha g/v} = \e^{-(\alpha g/v) \ln r}\), one sees that the localization length is exponentially large in~\(g^{-1}\).
Again, we see a fundamental, qualitative improvement under the additional assumption of commuting interactions.

\subsubsection{Enhanced LPPL for gapped ground states}

We also obtain a similar improvement for the \emph{local perturbations perturb locally} (LPPL) principle by adapting the complex-analytic subharmonicity argument from \textcite{WH2022}\@:
Let \(H\) be a long-range Hamiltonian with decay~\(\alpha\) as before and \(V_X\) a perturbation localized in~\(X\subset \Lambda\).
Moreover, assume that for all~\(\lambda\in \intervalcc{0,1}\), the ground state \(\rho_\lambda\) of \(H + \lambda \, V_X\) is gapped with gap at least~\(g\).
Then,
\begin{equation*}
    \abs[\Big]{
        \trace[\big]{\rho_0 \, B_Y}
        - \trace[\big]{\rho_1 \, B_Y}
    }
    \lesssim
    \norm{B_Y} \, \paren[\big]{\norm{V_X}+\norm{V_X}^2}
    \, \abs{X} \, \abs{Y}
    \, g^{-3} \, \dist{X,Y}^{-\alpha}
    .
\end{equation*}
The precise statement is given in \cref{thm:LPPL-for-gapped-ground-states}.
This bound becomes ballistic for \(\alpha=3\) in the same sense as before.
However, the best known bound for general interactions scales exactly the same as the bound for decay of correlations~\cite{WH2022} and thus the localization length is again exponentially large in~\(g^{-1}\).
A comparison of the different regimes and scalings is provided in \cref{table:app}.

\begin{table}
    \makeatletter
    \protected\def\savetpos#1{\pdfsavepos\write\@auxout{\gdef\string\tpos#1{\the\pdflastxpos}}}%
    \def\tposdifference#1#2{\ifcsname tpos#1\endcsname\the\dimexpr\csname tpos#2\endcsname sp -\dimexpr\csname tpos#1\endcsname sp\relax\else0.5\textwidth\fi}
    \makeatother
    \setcellgapes{2pt}%
    \renewcommand\cellset{%
        \renewcommand\arraystretch{0.8}%
    }%
    \renewcommand{\cellalign}{t}
    \makegapedcells
    \newcommand{\fnciteref}[1]{{\footnotesize\cite{#1}}}
    \newcommand{\fnciterefs}[1]{{\footnotesize\cite{#1}}}
    \centering
        \caption{
            Comparisons between previous results and our improvements for commuting interactions.
            These are rough summaries, and we refer to the theorems for the precise statements.
            In particular, we drop logarithmic corrections in \(r\) and, for decay of correlations and LPPL, we focus on the scaling as \(g \to 0\), and we drop polynomial prefactors depending on the gap \(g\).
        }
        \label{table:app}
        \sffamily
        \begin{tabular}{cccc}
            \toprule
            \savetpos{a}\hfill type \hfill
             & interaction decay
             & commutator LRB
             & \hfill decay of correlations/LPPL \hfill \savetpos{b}
            \\ \midrule
            long-range
             & \(r^{-\alpha}\), \(\alpha>D\)
             & \(\e^{v\abs{t}} \, r^{-\alpha}\)~\fnciteref{HK2006}
             & \(r^{-\alpha g/v}\)~\fnciteref{WH2022}
            \\
             & \makecell[cc]{\(r^{-\alpha}\), \(\alpha>2D\) }
             & \(t^{2D+1} \, (r-vt)^{-\alpha}\)~\fnciteref{KS2020}\textsuperscript{a}
             & \(r^{-\alpha}\)~\fnciterefs{TCE2020,WH2022}
            \\
            \makecell{commuting \\[-2pt] long-range}
             & \(r^{-\alpha}\), \(\alpha>0\)
             & \(\abs{t} \, r^{-\alpha}\)
             & \(r^{-\alpha}\)
            \\ \midrule
            short-range
             & \(\e^{-b r}\)
             & \(\e^{b'\,(v\abs{t} - r)}\)~\fnciteref{HK2006}
             & \(\e^{-g/v \, r}\)~\fnciterefs{HK2006,WH2022}
            \\ \addlinespace
            \makecell{commuting \\[-2pt] short-range}
             & \(\e^{-b r^p}\)
             & \(\abs{t} \, \e^{-b' r^p}\)
             & \(\e^{-b' r^p}\)
            \\ \bottomrule
            \multicolumn{4}{l}{
                \textsuperscript{a}\parbox[t]{\tposdifference{a}{b}}{\footnotesize%
                    This bound is valid only for \(\alpha>2D+1\).
                    For \(\alpha\in \intervaloo{2D,2D+1}\) one has \(t^{\frac{\alpha-D}{\alpha-2D}} \, r^{\alpha-D}\) from~\cite{TGB2021}, see \cref{fig:intro-commutator-version}.
                    Both focus on the shape of the light cone, while the proof of decay of correlations in~\cite{WH2022} relies on a bound with better decay in \(r\) proven in~\cite{FGC2015}.
                }}
        \end{tabular}
\end{table}

\subsubsection{A stability result}
Additionally, in \cref{sec:LR-bound-interaction-picture}, we prove a result on \emph{stability of LRBs} for more general Hamiltonians which are only assumed to have a commuting part, and we prove a LRB that is independent of the strength of the commuting part.
Stability-type results for LRBs have been proved in other settings before, e.g.\ in~\cite{NSY2019,TGDL2021,KS2020,TB2024}.

\medskip

This finishes our presentation of the main results.
We discuss future directions in \cref{sect:conclusions}.

\subsection{Motivation from quantum error-correcting codes}
\label{sec:motivation}
The results described above show that there is a fundamental, qualitative difference between the quantum many-body dynamics produced by general long-range interactions and commuting long-range interactions.

In this paragraph, we briefly review why Hamiltonians with commuting non-local interactions have recently received a lot of interest in the quantum information and quantum computing communities in the context of quantum error-correcting codes.
The goal of a quantum error-correcting code is typically to robustly store quantum information, often by harnessing topological order; early examples are the toric code~\cite{Kitaev2003} and CSS stabilizer codes~\cite{CS1996,Steane1996}.
Most quantum error-correcting codes are based on \emph{commuting} interactions.
Non-local \emph{and} commuting interactions are of particular interest lately because they arise in Euclidean constructions of efficient quantum codes.
The reason is as follows: Recall that a \([[n,k,d]]\) quantum code is a quantum code on system size \(n\) that can store \(k\) logical qubits with code distance \(d\).
The goal is to have \(k\) and \(d\) as large as possible, so that the code is both efficient and robust.
Then, the Bravyi-Poulin-Terhal (BPT) theorem~\cite{BPT2010} says that any \([[n,k,d]]\) quantum code with short-range interactions in \(D\) Euclidean dimensions must satisfy the bound
\(k \, d\mkern1mu^{2/(D-1)} \lesssim n\).
An extension to general graphs was given in~\cite{BK2022connectivity}.
A quantum code is called good if it satisfies \(k\),~\(d\sim n\) (which is optimal) and the BPT theorem says that local commuting interactions in \(D\) Euclidean dimensions cannot produce good quantum codes.
Then, in 2021, the 30-year-old open problem of constructing good quantum codes was resolved by quantum low-density parity check (qLDPC) codes on expander graphs and related constructions~\cite{PK2022,BE2021,LZ2022}.
Implementing these codes as Hamiltonians in the Euclidean setting, naturally produces Hamiltonians with non-local, commuting interactions~\cite{DBT2021,BK2022, HMKL2023,PJBP2024,BGKL2024} (as expected in view of the BPT theorem~\cite{BPT2010}).
Beside their potential for near-term fault-tolerant quantum computation, quantum codes are also emerging as a rich source of new models for condensed-matter physics, e.g., they lead to a proof of the no low-energy trivial state (NLTS) conjecture of Freedman and Hastings~\cite{FH2014,Hastings2013} and provide a rich class of gapped topological quantum phases~\cite{RKLO2024}.

For illustration purposes of the kind of quantum code that is within our scope, we discuss a simple example of a long-range toric code below.
Of course, Euclidean implementations of good qLDPC codes coming from expander graphs would look more complicated.
Indeed, the pioneering investigations for Euclidean implementations of efficient qLDPC codes and related error-correcting codes~\cite{DBT2021,BK2022, HMKL2023,PJBP2024,BGKL2024} have not fully settled on a specific Hamiltonian and non-locality in these works is not necessarily the same as power-law decay.
In some cases it should rather be understood as few-body and not necessarily decaying in distance.
Nonetheless, these proposals provide strong motivation for us to study quantum dynamics of Hamiltonians with commuting long-range interactions like the following one.

\begin{example}[Long-range toric code]
    \label{ex:lrtc}
    A non-trivial example of a Hamiltonian satisfying our assumptions for any \(\alpha>0\) is given by a long-range version of Kitaev’s famous toric code model~\cite{Kitaev2003}.
    The basic model consists of spins sitting at the edges of the lattice \((\Z/L\Z)^2\).
    We denote the set of vertices~\(V\), edges~\(E\) and faces~\(F\) of \((\Z/L\Z)^2\).
    Then the Hilbert space is given by \(\Hi := \bigotimes_{e\in E} \Hi_e := \bigotimes_{e\in E} \C^2\). 
    To each vertex \(s\in V\) one associates a \emph{star} operator \(A_s := -\unit + \prod_{e\in E: v\in e} \pauliX_e\), and to each face \(p\in F\) one associates a \emph{plaquette} operator \(B_p := -\unit + \prod_{e\in \partial (p)} \pauliZ_e\), where \(\pauli{\#}_e\) is the Pauli~\(\#\) matrix acting on~\(\Hi_e\) and \(\partial (p)\) denotes the four edges touching \(p\). 
    The Hamiltonian is then given as a sum
    \begin{equation*}
        H
        :=
        - \sum_{s \in V} A_s
        - \sum_{p \in F} B_p
        ,
    \end{equation*}
    which turns out to be commuting and gapped~\cite{Kitaev2003} and has ground state energy~\(0\).
    Moreover, it is frustration free, i.e.\ its ground state projection~\(P\) also minimizes all individual terms \(A_s \, P = 0\) and \(B_p \, P = 0\).
    Adding a term of the form
    \begin{equation*}
        H_{\textrm{long-range}}
        :=
        \sum_{s_1,s_2\in V} f(s_1,s_2) \, A_{s_1} \, A_{s_2}
        + \sum_{p_1,p_2\in F} g(p_1,p_2) \, B_{p_1} \, B_{p_2}
        ,
    \end{equation*}
    with \(0 \leq f(s_1,s_2) \leq C \, \dist{s_1,s_2}^{-\alpha}\) and \(0 \leq g(p_1,p_2) \leq C \, \dist{p_1,p_2}^{-\alpha}\) which is positive and satisfies \(H_{\textrm{long-range}} \, P = 0\) can only increase the gap.
    The perturbed Hamiltonian \(H+H_{\textrm{long-range}}\) is still mutually commuting, gapped, has ground state~\(P\) and is long-range.
    This example can be generalized in various ways, e.g.\ to higher dimensions.
    One can also add higher order products \((-1)^k \, A_{s_1} \dotsm A_{s_k}\) and \((-1)^k \, B_{p_1} \dotsm B_{p_k}\) or allow for small negative~\(f\) and~\(g\) and still obtain a spectral gap by adopting the BHM strategy~\cite{BHM2010,NSY2020}.
    Having a spectral gap is not needed for the Lieb-Robinson bounds.
    It is only relevant for the applications to gapped ground states that we discussed in \cref{sec:intro-further-results}.
\end{example}

\section{Mathematical setup}
\label{sec:setup}

In this work we want to consider spin systems on \(D\)-regular graphs.
The graph is seen as a metric space \((\Gamma,d)\), where \(d\) is the graph distance.
For any \(Z\subset \Gamma\) we denote its cardinality by \(\abs{Z}\) and its diameter by \(\diam{Z} := \sup_{x, y \in Z} d(x,y)\).
Given two sets \(X,Y\subset \Gamma\) we denote by \(\dist{X,Y}\) their distance with respect to the metric \(d\).
We use \(\Lambda\Subset\Gamma\) to denote that \(\Lambda\) is a finite subset of \(\Gamma\).
The ball with radius \(r\) around \(x\in \Gamma\) is denoted \(B^\Gamma_x(r) := \Set{z\in \Gamma \given d(x,z) \leq r}\).

\begin{definition}[Surface-regular graph~\cite{RS2024}]
    \label{defn:graph}
    We call the graph \((\kappa,D)\)-surface-regular, if
    \begin{equation*}
        \abs[\big]{B^\Gamma_x(r) \setminus B^\Gamma_x(r-1)}
        \leq
        \kappa \, r^{D-1}
        \qquadtext{for all}
        x\in \Gamma
        \quadtext{and}
        r\in \N
        .
        \qedhere
    \end{equation*}
\end{definition}

The prototypical example of a \((\kappa,D)\)-surface-regular lattice is \(\Gamma=\Z^D\).
Any \(D\)-surface-regular graph is also \(D\)-regular~\cite{RS2024}, in the sense that
\begin{equation*}
    \abs[\big]{B^\Gamma_x(r)}
    \leq
    (1+\kappa) \, r^{D}
    \qquadtext{for all}
    x\in \Gamma
    \quadtext{and}
    r\in \N_{>0}
    .
\end{equation*}
This notion of \(D\)-regular graphs is more common in the literature.
However, all lattices we are interested in are \(D\)-surface-regular, and we can prove better bounds with this stronger assumption.
An explicit construction of a counter example can be found in~\cite{Tessera2007}.

Often however, we restrict to \emph{finite} \((\kappa,D)\)-surface-regular lattices \(\Lambda\).
This might seem trivial, but it is important to note that all our results and constants will be independent of the specific graph \(\Lambda\) and only depend on \(\kappa\) and \(D\).
In certain cases, this will allow us to take the limit \(\Lambda_n \to \Gamma\) where all \(\Lambda_n\) are finite \((\kappa,D)\)-regular lattices and \(\Gamma\) is an infinite one.

With every site \(x \in \Gamma\) one associates a finite-dimensional local Hilbert space \(\Hi_x := \C^q\) with the corresponding space of linear operators denoted by \(\alg_x:=\mathcal{B} \paren*{\C^q}\).
And for every finite \(\Lambda \Subset \Gamma\) we define the Hilbert space \(\Hi_{\Lambda}:= \bigotimes_{x \in \Lambda} \Hi_x\), and denote the algebra of bounded linear operators on~\(\Hi_{\Lambda}\) by \(\alg_{\Lambda} := \mathcal{B}(\Hi_{\Lambda})\).
Due to the tensor product structure, we have \(\alg_{\Lambda} = \bigotimes_{x \in \Lambda} \alg_x\).
Hence, for \(X \subset \Lambda \Subset \Gamma\), any \(A \in \alg_{X}\) can be viewed as an element of~\(\alg_{\Lambda}\) by identifying~\(A\) with \(A \otimes \unit_{\Lambda \backslash X} \in \alg_{\Lambda}\), where~\(\unit_{\Lambda \backslash X}\) denotes the identity in~\(\alg_{\Lambda \backslash X}\).
This identification is always understood implicitly and for \(B\in \alg_\Lambda\) we denote by~\(\supp(B)\) the smallest \(Y\subset \Lambda\) such that~\(B\in \alg_Y\).
Moreover, the algebra of local observables on \(\Gamma\) is given by
\begin{equation*}
    \algloc_{\Gamma}
    :=
    \bigcup_{\Lambda \Subset \Gamma} \alg_{\Lambda}
    ,\qquadtext{and its completion}
    \alg_\Gamma
    :=
    \overline{\algloc_\Gamma}^{\raisebox{-.2ex}{\kern1pt\(\scriptstyle\norm{\cdot}\)}}
\end{equation*}
with respect to the operator norm is the algebra of \emph{quasi-local} observables, which is only relevant if \(\Gamma\) is infinite.

An \emph{interaction} on \(\Gamma\) is a function
\begin{equation}
    \Psi \colon \List{Z\Subset\Gamma} \to \alg_{\Gamma},
    \quad Z\mapsto \Psi (Z)\in\alg_Z
    \quadtext{with} \Psi(Z)=\Psi(Z)^*
    .
\end{equation}
For each \(\Lambda \Subset \Gamma\), the corresponding local Hamiltonian is then defined as
\begin{equation*}
    H_{\Lambda}
    :=
    \sum_{Z \subset \Lambda} \Psi(Z)
    .
\end{equation*}
And the Heisenberg time-evolution of an operator \(A\in \alg_\Lambda\) is denoted by
\begin{equation}
    \label{eq:definition-time-evolution}
    \tau_t^{\Lambda}(A)
    :=
    \e^{\I t H_{\Lambda}} \, A \, \e^{-\I t H_{\Lambda}}
    .
\end{equation}
For infinite \(\Gamma\), the Hamiltonian is not an object of the algebra \(\alg_\Gamma\) and the dynamics cannot be defined as in~\eqref{eq:definition-time-evolution}.
Instead, one can define the time evolution on finite subgraphs \(\Lambda \Subset \Gamma\) as above.
The dynamics on~\(\Gamma\) can then be defined as the limit
\begin{equation}
    \label{eq:definition-thermodynamic-limit}
    \tau^\Gamma_t(A)
    :=
    \lim_{\Lambda\nearrow\Gamma} \tau^\Lambda_t(A)
\end{equation}
if it exists for all \(A\in \algloc_\Gamma\).
It might then be extended to a cycle of automorphism on~\(\alg_\Gamma\).
This limit is understood in the following way:
Let \(\List{\Lambda_n}_{n\in \N}\subset \Set{\Lambda\Subset\Gamma}\) be an increasing and exhaustive sequences, i.e.\ \(\Lambda_n\subset \Lambda_{n+1}\) for all \(n\) and for each \(Z\Subset\Gamma\) there is \(n\in \N\) such that \(Z\subset \Lambda_n\).
If the limit \(\lim_{n\to\infty} \tau^{\Lambda_n}_t(A)\) exists and is independent of the chosen sequence, it equals the above limit.
We comment on the existence after introducing interaction norms.

To measure locality of an interaction so-called interaction norms are widely used in the literature.
Therefore, let \(F\colon \R_{\geq 0} \to \R_{>0}\) be a decaying function.
We define (for interactions on \(\Gamma\))
\begin{equation}
    \label{eq:def-interaction-norm}
    \norm{\Psi}_F
    :=
    \sup_{x,y\in \Gamma} \sumstack{Z\Subset\Gamma \suchthat\\ x,y\in Z}
    \frac
        {\norm{\Psi(Z)}}
        {F\pd{x,y}}
    ,
\end{equation}
and say that \(\Psi\) is \(F\)-local if \(\norm{\Psi}_F < \infty\).
We choose to use this norm, because it allows for very general interactions.
For two-body interactions with decay \(\norm{\Psi\paren[\big]{\List{x,y}}} \leq C \, F\pdist{x,y}\), one directly computes \(\norm{\Psi}_F \leq C\).

Our bounds will depend on a suitable interaction norm \(\norm{\Psi}_F\).
Beyond this, however, the constants in all our results will be independent of the particular lattice~\(\Gamma\).
In this way, if one defines an interaction on an infinite graph~\(\Gamma\) such that \(\norm{\Psi}_F^{(\Gamma)}<\infty\), one obtains bounds that are uniform in~\(\Lambda \Subset \Gamma\), by using \(\norm{\Psi}_F^{(\Lambda)} \leq \norm{\Psi}_F^{(\Gamma)}\).
Here, we add an extra index to highlight on which graph~\eqref{eq:def-interaction-norm} is used.
This in particular holds for the prototypical example \(\Gamma=\Z^D\), where one requires
\begin{equation*}
    \norm{\Psi}^{(\Z^D)}_F
    :=
    \sup_{x,y\in \Z^D} \sumstack{Z\Subset\Z^D \suchthat\\ x,y\in Z}
    \frac
        {\norm{\Psi(Z)}}
        {F\pd{x,y}}
\end{equation*}
and uniform obtains bounds on arbitrary subsets \(\Lambda\Subset\Z^D\).
But with the more general framework, one can for example also consider rectangles with torus geometry, or more general graphs like the honeycomb lattice.

With the notion of interaction norms at hand, we can come back to the thermodynamic limit.
For interactions with decay \(F(r) \leq (1+r)^{-\alpha}\) with \(\alpha>D\) on a \(D\)-surface-regular graph, existence of the thermodynamic limit is known, see e.g.~\cite[Theorem~3.5]{NSY2019}.%
\footnote{
    To apply their results, one needs to realize that \(F_\alpha\) satisfies the convolution condition~\cite[equation~3.9]{NSY2019} on a \(D\)-surface-regular graph \(\Gamma\) if \(\alpha>D\), by the same arguments as given in~\cite[section~8.1.1]{NSY2019} for \(\Gamma=\Z^D\).
    Alternatively, one can just use~\eqref{eq:thm-local-approximation-of-time-evolution-final} as outlined in the main text.
}
For slower decaying interactions, one does not expect a thermodynamic limit to exist in general.
In fact, one can expect that the interaction at least need to be summable in the sense
\begin{equation*}
    \sup_{z\in \Gamma} \, \sumstack[lr]{Z\Subset \Gamma\suchthat\\z\in Z} \, \norm{\Psi(Z)}
    <
    \infty
    ,
\end{equation*}
to take a thermodynamic limit.
See for example~\cite[Theorem~7.6.2]{Ruelle1969} for the existence of the thermodynamic limit without explicitly requiring decay in the diameter~\(\diam{Z}\) of the supports.

To discuss the operator localization version of the LRB, we introduce the conditional expectation \(\cexp{Y}{} \colon \alg_\Lambda \to \alg_Y\), which has many crucial properties, which are for example discussed in~\cite[Lemma~4.1]{NSY2019}.
For us, it is only important to know that \(\norm{\cexp{Y}{}} \leq 1\) and that its restriction to \(\alg_Y\) is the identity.
For finite~\(\Lambda\) it is just given by the partial trace on the complement~\cite{BHV2006}
\begin{equation}
    \label{eq:definition-conditional-expectation}
    \cexp{Y}{A}
    =
    \frac
        {1}
        {\trace_{\Lambda\setminus Y}{\unit_{\Lambda\setminus Y}}}
    \, \trace_{\Lambda\setminus Y}{A}
    ,
\end{equation}
and a generalization to infinite volume systems, which we will use in some of the statements, are defined in~\cite[Lemma~4.2]{NSY2019}.
Extensions to fermionic systems~\cite{NSY2017} also exist, and the LRBs we discuss here similarly apply for fermions.

\section{Lieb-Robinson bounds for commuting interactions}
\label{sec:LR-bounds-for-commuting-Hamiltonians}

We now restrict to mutually commuting interactions satisfying \(\commutator[\big]{\Psi(X), \Psi(Y)} = 0\) for all \(X\) and~\(Y\).
We split the discussion according to the decay of the interactions.

\subsection{Finite-range interactions}
\label{sec:LR-bound-finite-range}

The simplest type of interactions are finite-range interactions, i.e.\ interactions for which there exists \(R\) such that \(\Psi(Z)=0\) if \(\diam{Z}>R\).
For commuting, finite-range interactions, the Heisenberg time evolution remains localized for all times in the following sense.

\begin{theorem}[Lieb-Robinson bound for commuting, finite range Hamiltonians]
    \label{cor:LR-bound-commuting-finite-range}
    Let \(\Psi\) be a commuting interaction on a graph \(\Gamma\).
    Assume that \(\Psi\) is of range \(R\), i.e.\ \(\Psi(Z)=0\) if \(\diam{Z}>R\), and uniformly bounded, i.e.\ there exists a constant \(C_\Psi\) such that \(\norm{\Psi(Z)}<C_\Psi\) for all \(Z\Subset \Gamma\).

    Then the corresponding Heisenberg time-evolution satisfies
    \begin{equation*}
        \tau^\Gamma_t(A) = \tau^{X_R}_t(A) \in \alg_{X_R}
        \qquadtext{for all}
        X\Subset \Gamma
        \quadtext{and}
        A\in \alg_X
        ,
    \end{equation*}
    where \(X_R := \Set{z\in \Gamma \given \dist{z,X} \leq R}\) is the \(R\) neighbourhood of \(X\).
    Thus, it also satisfies the following Lieb-Robinson bound:
    For all \(X\), \(Y\Subset \Gamma\) and operators \(A\in \alg_X\), \(B\in \alg_Y\),
    \begin{equation*}
        \norm[\big]{
            \commutator{ \tau^\Gamma_t(A), B }
        }
        =
        0
        \qquadtext{if}
        \dist{X,Y} > R
        .
    \end{equation*}
\end{theorem}

This behaviour of commuting finite-range interactions was observed before, for example in~\cite{PHKM2010}.
We note that this behaviour dramatically differs from the general LRB for non-commuting, finite-range Hamiltonians.
Our goal will be to show that this extreme locality persists to some degree for long-range commuting interactions.

\subsection{Short-range interactions}
\label{sec:LR-bound-short-range}

Another class of widely used interactions are so-called short-range or (stretched) exponentially-decaying interactions, which have a bounded interaction norm with decay function
\begin{equation*}
    F_{b,p}(r)
    :=
    \e^{-b \, r^p}
\end{equation*}
for some \(b>0\) and \(p\in \intervaloc{0,1}\).

\begin{theorem}
    \label{thm:LR-bound-short-range}
    Let \(0<b'<b\), \(p\in \intervaloc{0,1}\), \(\kappa>0\) and \(D\in \N\), then there exists a constant \(C_{b,b'}>0\) such that the following holds:
    Let \(\Psi\) be a commuting interaction on a \((\kappa,D)\)-regular graph \(\Gamma\) satisfying \(\norm{\Psi}_{F_{b,p}}<\infty\).
    Then the corresponding Heisenberg time-evolution satisfies the following Lieb-Robinson bounds:
    \begin{thmlist}
        \item \label{item-thm:LR-bound-short-range:commutator}
            For all disjoint \(X\), \(Y\Subset \Gamma\) and operators \(A\in \alg_X\), \(B\in \alg_Y\),
            \begin{equation}
                \label{eq:short-range-LR-bound}
                \norm[\big]{
                    \commutator{ \tau^\Gamma_t(A), B }
                }
                \leq
                C_{b,b'} \, \norm{\Psi}_{F_{b,p}} \, \norm{A} \, \norm{B}
                \, \min\List[\big]{\abs{X}, \abs{Y}} \, \abs{t} \, F_{b',p}\pdist{X,Y}
                .
            \end{equation}
        \item \label{item-thm:LR-bound-short-range:operator-localization}
            For all \(X\Subset \Gamma\), operators \(A\in \alg_X\), and \(r \geq 0\),
            \begin{equation*}
                \norm[\big]{
                    \tau^\Gamma_t(A)
                    -\cexp[\big]{X_r}{\tau_t^{\Gamma}(A)}
                }
                \leq
                C_{b,b'} \, \norm{\Psi}_{F_{b,p}} \, \norm{A}
                \, \abs{X} \, \abs{t} \, F_{b',p}(r)
                .
            \end{equation*}
    \end{thmlist}
\end{theorem}

A more general result for short-range interactions, assuming only that the commutator of interaction terms is small (but not necessarily zero), was proven by \textcite{HHKL2021}.
In the limit of vanishing commutators, their bound gives the same LRB\@.
The result stated here, gives a logarithmic light cone, which \textcite{TB2024} call \enquote{slow dynamics}.
They then prove stability of the LRBs to local perturbations similar to~\eqref{eq:LR-bound-perturbed-Hamiltonian} but for extended perturbations given by an interaction.

\subsection{Long-range interactions}
\label{sec:LR-bound-long-range}

Recently, many efforts were made to prove LRBs for long-range interactions, which are described by the decay function
\begin{equation*}
    F_\alpha(r) := (1+r)^{-\alpha}
    .
\end{equation*}
While this is quite involved for general (non-commuting) Hamiltonians, one can very easily obtain LRBs for commuting Hamiltonians.
Moreover, it is often necessary to distinguish different regimes in \(\alpha\).
In particular, on a \(D\)-surface-regular graph, one cannot expect a thermodynamic limit of the dynamics for \(0<\alpha\leq D\), which is the reason for us to provide two distinct statements.
The first works also for infinite graphs and allows for two different bounds, one scales better in the support of the two observables, the other scales better in the distance between them.

\begin{theorem}
    \label{thm:LRB-long-range-alpha-ge-nu}
    Let \(\kappa>0\), \(D\in \N\) and \(\alpha>D\), then there exists a constant \(C>0\) such that the following holds:
    Let \(\Psi\) be a commuting interaction on a \((\kappa,D)\)-surface-regular graph \(\Gamma\) satisfying \(\norm{\Psi}_{F_\alpha}<\infty\).
    Then the corresponding Heisenberg time-evolution satisfies the following Lieb-Robinson bounds:
    \begin{thmlist}
        \item \label{item-thm:LRB-long-range-alpha-ge-nu:commutator}
            For all disjoint \(X\), \(Y\Subset \Gamma\) and operators \(A\in \alg_X\), \(B\in \alg_Y\),
            \begin{align}
                \label{eq:long-range-LR-bound-alpha-ge-nu}
                \norm[\big]{
                    \commutator{ \tau^\Gamma_t(A), B }
                }
                &\leq
                C \, \norm{\Psi}_{F_\alpha} \, \norm{A} \, \norm{B}
                \, \min\List[\big]{\abs{X}, \abs{Y}} \, \abs{t} \, F_{\alpha-D}\pdist{X,Y}
                ,
            \shortintertext{and}
                \label{eq:long-range-LR-bound-alpha-ge-0}
                \norm[\big]{
                    \commutator{ \tau^\Gamma_t(A), B }
                }
                &\leq
                4 \, \norm{\Psi}_{F_\alpha} \, \norm{A} \, \norm{B}
                \, \abs{X} \, \abs{Y} \, \abs{t} \, F_\alpha\pdist{X,Y}
                .
            \end{align}
        \item \label{item-thm:LRB-long-range-alpha-ge-nu:operator-localization}
            For all \(X\Subset \Gamma\), operators \(A\in \alg_X\), and \(r \geq 0\),
            \begin{equation}
                \label{eq:long-range-LR-bound-alpha-ge-nu-operator-localization}
                \norm[\big]{
                    \tau^\Gamma_t(A)
                    -\cexp[\big]{X_r}{\tau_t^{\Gamma}(A)}
                }
                \leq
                C \, \norm{\Psi}_{F_\alpha} \, \norm{A}
                \, \abs{X} \, \abs{t} \, F_{\alpha-D}(r)
                .
            \end{equation}
    \end{thmlist}

\end{theorem}

This gives a linear light cone (a ballistic LRB) in the sense of local approximations for \(\alpha>D+1\) which is an improvement over the tight \(\alpha > 2 D + 1\) found for non-commuting two-body interactions~\cite{TGB2021}.

But even for \(\alpha>0\), we can obtain a bound as in~\eqref{eq:long-range-LR-bound-alpha-ge-0} on arbitrarily large, but finite graphs.

\begin{theorem}
    \label{thm:LR-bound-long-range-alpha-ge-0}
    Let \(\alpha>0\) and \(\Lambda\) be a finite graph.
    Let \(\Psi\) be a commuting interaction on \(\Lambda\) satisfying \(\norm{\Psi}_{F_\alpha}<\infty\).
    Then the corresponding Heisenberg time-evolution satisfies the Lieb-Robinson bound given in~\eqref{eq:long-range-LR-bound-alpha-ge-0}.
\end{theorem}

This gives some kind of linear light cone for commuting long-range interactions with \(\alpha>1\), which are not even summable in dimensions \(D>1\).
In this parameter range, the Heisenberg dynamics do not converge to a thermodynamic limit in general.
Hence, a bound like~\eqref{eq:long-range-LR-bound-alpha-ge-nu} cannot hold, as it implies convergence of the Heisenberg dynamics in the thermodynamic limit.
Thus, we cannot discuss LRBs for the limiting dynamics on infinite graphs.

\begin{remark}
    \label{remark:optimality}
    The Lieb-Robinson bounds in \cref{thm:LRB-long-range-alpha-ge-nu} are sharp.
    For example, to see that~\eqref{eq:long-range-LR-bound-alpha-ge-nu} is sharp, let \(X = \List{x}\), \(Y = \List{y}\) and choose the interaction with only non-vanishing term \(\Psi(\List{x,y}) = F_\alpha\pdist{x,y} \, U\), where \(U=\e^{\I\frac{\pi}{4}(\unit_x+\pauliZ_x)(\unit_y-\pauliX_y)}\) is a standard CNOT gate between sites \(x\) and \(y\).
    The CNOT gate \(U\) satisfies \(\e^{\I t U} = \cos(t) \, \unit + \I \sin(t) \, U\) and \(\norm{U} = 1\).
    Moreover, let \(A=\pauliX_x\) be a spin flip and \(B=\pauliZ_y\) measure the \(\componentZ\)-component of the qubit at site \(y\).
    Now, let \(\psi\) be any pure state on \(\Lambda\setminus\List{x,y}\).
    Then
    \begin{equation*}
        \commutator{ \tau^\Lambda_t(A), B }
        \, \ket{\downarrow \downarrow}_{xy} \otimes \psi
        =
        - 2 \, \I \sin\paren[\big]{t \, F_\alpha\pdist{x,y}}
        \, \ket{\uparrow \uparrow}_{xy} \otimes \psi
        ,
    \end{equation*}
    and thus
    \begin{equation*}
        \norm[\big]{
            \commutator{ \tau^\Lambda_t(A), B }
        }
        \geq
        2 \, \abs[\big]{
            \sin\paren[\big]{t \, F_\alpha\pdist{x,y}}
        }
        ,
    \end{equation*}
    which for small \(t \, F_\alpha\pdist{x,y}\) has the same scaling as~\eqref{eq:long-range-LR-bound-alpha-ge-0}.
    In \cref{sec:optimality-of-the-LRB}, we present a generalization of this example that extends sharpness of~\eqref{eq:long-range-LR-bound-alpha-ge-0} to the case of general supports \(X\), \(Y\) and we also show sharpness of the other bounds from \cref{thm:LRB-long-range-alpha-ge-nu}.
\end{remark}

\subsection{Proof of the Lieb-Robinson bound for commuting Hamiltonians}

The proof of all the stated LRBs boils down to a very simple observation, which was already made before, e.g.\ in~\cite{PHKM2010}.
It drastically simplifies the Heisenberg evolution of local operators by only considering the interaction terms, which overlap with the support of the operator.

\begin{lemma}\label{lem:evolution-interaction-intersecting-support}
    Let \(\Psi\) be a commuting interaction on a finite graph \(\Lambda\).
    For all \(X \subset \Lambda\) and operators \(A\in \alg_X\), the corresponding Heisenberg time-evolution satisfies
    \begin{equation*}
        \tau_t^{\Lambda}(A)
        =
        \tau_{t,{\cap X}}^{\Lambda}(A)
        ,
    \end{equation*}
    where \(\tau_{t,{\cap X}}^{\Lambda}\) is generated by the interaction
    \begin{equation}
        \label{eq:interaction-intersecting-X}
        \Psi_{\cap X}(Z)
        :=
        \begin{cases*}
            \Psi(Z) & if \(X\cap Z\neq \emptyset\) and \\
            0       & else.
        \end{cases*}
    \end{equation}
\end{lemma}

\begin{proof}
    Due to commutativity of the interaction, the exponential satisfies
    \begin{equation*}
        \e^{\I t H_{\Lambda}}
        =
        \e^{\I t H_{\Lambda \setminus X}}
        \, \e^{\I t (H_{\Lambda} - H_{\Lambda \setminus X})}
        .
    \end{equation*}
    Using this for both exponentials in~\eqref{eq:definition-time-evolution} and realizing that \(\commutator[\big]{H_{\Lambda \setminus X}, A} = 0\) due to their disjoint support, one has
    \begin{equation*}
        \tau_t^{\Lambda}(A)
        =
        \e^{-\I t (H_{\Lambda} - H_{\Lambda \setminus X})} \, A \, \e^{\I t (H_{\Lambda} - H_{\Lambda \setminus X})}
        =
        \tau_{t,{\cap X}}^{\Lambda}(A)
        .
        \qedhere
    \end{equation*}
\end{proof}

From here, the statement about finite-range Hamiltonians on finite graphs in \cref{cor:LR-bound-commuting-finite-range} is immediately clear.
The extension to infinite graphs follows from the next theorem, and we comment on the details at the end of this section.

\Cref{lem:evolution-interaction-intersecting-support} also allows to approximate the time-evolution locally.

\begin{theorem}[Local approximation of Heisenberg evolution]
    \label{thm:local-approximation-of-time-evolution}
    Let \(\Psi\) be a commuting interaction on a finite graph \(\Lambda\).
    For all \(X \subset \Lambda' \subset \Lambda\) and \(A\in \alg_X\) it holds that
    \begin{equation*}
        \norm[\big]{
            \tau_t^\Lambda(A)
            - \tau_t^{\Lambda'}(A)
        }
        \leq
        2 \, \norm{A} \, \abs{t}
        \, \sumstack[lr]{
            Z \subset \Lambda \suchthat\\
            Z \cap X \neq \emptyset,\\
            Z \cap \Lambda\setminus\Lambda' \neq \emptyset
        } \, \norm[\big]{\Psi(Z)}
        .
    \end{equation*}
    Moreover, if \(\norm{\Psi}_F<\infty\) for some decaying function \(F\colon \R_{\geq 0} \to \R_{>0}\), then
    \begin{equation}
        \label{eq:thm-local-approximation-of-time-evolution-final}
        \norm[\big]{
            \tau_t^\Lambda(A)
            - \tau_t^{\Lambda'}(A)
        }
        \leq
        2 \, \norm{A} \, \norm{\Psi}_F \, \abs{t}
        \sum_{x\in X} \sumstack[r]{y\in \Lambda\setminus\smash{\cramped[\scriptstyle]\Lambda'}} \, F\pd{x,y}
        .
    \end{equation}
\end{theorem}

\begin{proof}
    By~\cref{lem:evolution-interaction-intersecting-support} we have
    \begin{equation*}
        \begin{aligned}
            \tau_t^{\Lambda}(A)
            -\tau_t^{\Lambda'}(A)
            &=
            \tau_{t,{\cap X}}^{\Lambda}(A)
            -\tau_{t,{\cap X}}^{\Lambda'}(A)
            \\&=
            \int_0^t \odv*{}{s} \, \tau_{s,{\cap X}}^{\Lambda} \circ \tau_{t-s,{\cap X}}^{\Lambda'}(A) \diff s
            \\&=
            \sumstack[l]{Z\subset \Lambda \suchthat\\ Z\cap \Lambda\setminus\Lambda' \neq \emptyset}
            \int_0^t \tau_{s,{\cap X}}^{\Lambda} \, \paren[\Big]{
                \commutator[\Big]{
                    \Psi_{\cap X}(Z),
                    \tau_{t-s,{\cap X}}^{\Lambda'}(A)
                }
            } \diff s
            ,
        \end{aligned}
    \end{equation*}
    and thus the statement after bounding
    \begin{equation*}
        \norm[\bigg]{
            \int_0^t \tau_{s,{\cap X}}^{\Lambda} \, \paren[\Big]{
                \commutator[\Big]{
                    \Psi_{\cap X}(Z),
                    \tau_{t-s,{\cap X}}^{\Lambda'}(A)
                }
            } \diff s
        }
        \leq
        2 \, \abs{t} \, \norm{\Psi_{\cap X}(Z)} \, \norm{A}
        .
    \end{equation*}
    For the second statement, we additionally bound
    \begin{equation*}
        \sumstack[lr]{
            Z \subset \Lambda \suchthat\\
            Z \cap X \neq \emptyset,\\
            Z \cap \Lambda\setminus\Lambda' \neq \emptyset
        } \, \norm[\big]{\Psi(Z)}
        \leq
        \sum_{x\in X} \sumstack[r]{y\in \Lambda\setminus\smash{\cramped[\scriptstyle]\Lambda'}} \, F\pd{x,y}
        \, \sumstack[lr]{Z\subset \Lambda \suchthat\\ x,y\in Z} \frac{\norm[\big]{\Psi(Z)}}{F\pd{x,y}}
        \leq
        \norm{\Psi}_F
        \sum_{x\in X} \sumstack[r]{y\in \Lambda\setminus\smash{\cramped[\scriptstyle]\Lambda'}} \, F\pd{x,y}
    \end{equation*}
    by overcounting.
\end{proof}

As a direct consequence, we also obtain the two versions of the LRB\@.

\begin{corollary}[Lieb-Robinson bound for commuting Hamiltonians]
    \label{cor:LRB-commuting-interaction-norm}
    Let \(\Psi\) be a commuting interaction on a finite graph \(\Lambda\).
    Let \(F\colon \R_{\geq 0} \to \R_{>0}\) be a decaying function and assume that \(\Psi\) is \(F\) local in the sense that \(\norm{\Psi}_F<\infty\).
    Then the corresponding Heisenberg time-evolution satisfies the following Lieb-Robinson bounds:
    \begin{thmlist}
        \item \label{item-cor:LRB-commuting-interaction-norm:commutator}
            For all disjoint \(X\), \(Y\subset \Lambda\) and operators \(A\in \alg_X\), \(B\in \alg_Y\),
            \begin{equation*}
                \norm[\big]{
                    \commutator{ \tau^\Lambda_t(A), B }
                }
                \leq
                4 \, \norm{\Psi}_F \, \norm{A} \, \norm{B} \, \abs{t}
                \, \sum_{x\in X} \sum_{y\in Y} F\pd{x,y}
                .
            \end{equation*}
        \item \label{item-cor:LRB-commuting-interaction-norm:operator-localization}
            For all \(X \subset \Lambda\), operators \(A\in \alg_X\), and \(r \geq 0\),
            \begin{equation*}
                \norm[\big]{
                    \tau^\Lambda_t(A)
                    -\cexp[\big]{X_r}{\tau_t^{\Lambda}(A)}
                }
                \leq
                4 \, \norm{\Psi}_F \, \norm{A} \, \abs{t}
                \, \sum_{x\in X} \sum_{y\in \Lambda\setminus X_r} F\pd{x,y}
                .
            \end{equation*}
    \end{thmlist}
\end{corollary}

\begin{proof}
    Both statements directly follow from \cref{thm:local-approximation-of-time-evolution}.
    For~\itemref{item-cor:LRB-commuting-interaction-norm:commutator} choose \(\Lambda' = \Lambda \setminus Y\).
    Then \(\tau^{\Lambda \setminus Y}_t(A)\in \alg_{\Lambda \setminus Y}\) and \(B\in \alg_Y\) have disjoint support and thus commute.
    Hence,
    \begin{equation*}
        \norm[\big]{
            \commutator{ \tau^\Lambda_t(A), B }
        }
        \leq
        \norm[\big]{
            \commutator{ \tau^\Lambda_t(A) - \tau^{\Lambda \setminus Y}_t(A), B }
        }
        \leq
        2
        \, \norm[\big]{
            \tau^\Lambda_t(A) - \tau^{\Lambda \setminus Y}_t(A)
        }
        \, \norm{B}
        .
    \end{equation*}
    For~\itemref{item-cor:LRB-commuting-interaction-norm:operator-localization} choose \(\Lambda'=X_r\).
    Then using \(\cexp{X_r}{\tau_t^{X_r}(A)} = \tau_t^{X_r}(A)\) and \(\norm{\cexp{X_r}{}}=1\) we obtain
    \begin{align*}
        \norm[\big]{
            \tau^\Lambda_t(A)
            -\cexp[\big]{X_r}{\tau_t^{\Lambda}(A)}
        }
        &\leq
        \norm[\big]{
            \tau^\Lambda_t(A)
            -\tau^{X_r}_t(A)
        }
        +
        \norm[\big]{
            \cexp[\big]{X_r}{\tau^{X_r}_t(A)}
            -\cexp[\big]{X_r}{\tau_t^{\Lambda}(A)}
        }
        \\&\leq
        2 \, \norm[\big]{
            \tau^\Lambda_t(A)
            -\tau^{X_r}_t(A)
        }
        .
        \qedhere
    \end{align*}
\end{proof}

Note, that the local approximation is different between \cref{thm:local-approximation-of-time-evolution} for \(\Lambda'=X_r\) and \cref{item-cor:LRB-commuting-interaction-norm:operator-localization}.
The first compares the evolution of~\(A\) to the time evolution generated by the local Hamiltonian on \(X_r\), while the latter gives an approximation of \(\tau_t^\Lambda(A)\) in~\(X_r\) by the conditional expectation.
The former implies the latter.

We can now prove~\cref{thm:LR-bound-short-range,thm:LRB-long-range-alpha-ge-nu,thm:LR-bound-long-range-alpha-ge-0} on finite graphs.
The extension to infinite graphs will be discussed afterwards.
Notice that
\begin{equation}
    \label{eq:trivial-bound-double-sum-F}
    \sum_{x\in X} \sum_{y\in Y} F\pd{x,y}
    \leq
    \abs{X} \, \abs{Y} \, F\pdist{X,Y}
\end{equation}
for any \(F\).
Together with \cref{cor:LRB-commuting-interaction-norm}, this proves \cref{thm:LR-bound-long-range-alpha-ge-0} and~\eqref{eq:long-range-LR-bound-alpha-ge-0} from \cref{item-thm:LRB-long-range-alpha-ge-nu:commutator}.
Alternatively, one can bound
\begin{equation*}
    \sum_{x\in X} \sum_{y\in Y} F\pd{x,y}
    \leq
    \min\List[\big]{\abs{X},\abs{Y}}
    \, \sup_{x\in \Lambda}
    \ \sumstack[lr]{y\in \Lambda\suchthat\\ \dd{x,y} \geq \dist{X,Y}}
    \, F\pd{x,y}
\end{equation*}
if the right-hand side is bounded.
Indeed, for the decay function \(F_\alpha(r) := (1+r)^{-\alpha}\),
a simple integration shows that for all \(\kappa>0\), \(D\in \N\), \(\alpha>D\), there exists a constant \(C/4>0\) such that
\begin{equation}
    \label{eq:bound-sum-F-polynomial-decay}
    \sup_{x\in \Lambda}
    \ \sumstack[lr]{y\in \Lambda \suchthat\\ d(x,y) \geq R}
    F_\alpha\pd{x,y}
    \leq
    \frac{C}{4} \, F_{\alpha-D}(R)
\end{equation}
for all \((\kappa,D)\)-surface-regular graphs \(\Lambda\).
This proves the remaining bounds from \cref{thm:LRB-long-range-alpha-ge-nu}.

Similarly, for \cref{thm:LR-bound-short-range}, we observe that for all \(0<b'<b\), \(p\in \intervaloc{0,1}\), \(\kappa>0\) and \(D\in \N\), there exists a constant \(C>0\) such that
\begin{equation}
    \label{eq:bound-sum-F-stretched-exponential-decay}
    \sup_{x\in \Lambda}
    \ \sumstack[lr]{y\in \Lambda \suchthat\\ \dd{x,y} \geq R}
    \, F\pd{x,y}
    \leq
    C\, \e^{-b'R^p}
\end{equation}
for all \((\kappa,D)\)-regular graphs.

It is left to obtain the results also on infinite graphs \(\Gamma\).
As discussed in \cref{sec:setup}, we first define the Heisenberg dynamics on \(\algloc_\Gamma\) as the limit
\begin{equation*}
    \tau^\Gamma_t(A)
    :=
    \lim_{\Lambda\nearrow\Gamma} \tau^\Lambda_t(A)
    ,
\end{equation*}
which is well-defined if \(\norm{\Psi}_{F_\alpha}<\infty\) for some \(\alpha>D\).
To obtain the LRBs on \(\Gamma\) one uses triangle inequality in \cref{thm:local-approximation-of-time-evolution} to bound
\begin{equation*}
    \norm[\big]{
        \tau^\Gamma_t(A) - \tau^{X_r}_t(A)
    }
    \leq
    \norm[\big]{
        \tau^\Gamma_t(A) - \tau^{\Lambda}_t(A)
    }
    +
    \norm[\big]{
        \tau^\Lambda_t(A) - \tau^{X_r}_t(A)
    }
    ,
\end{equation*}
where the first term vanishes in the limit \(\Lambda\nearrow\Gamma\) and the second gives the previously proven LRBs which did not depend on \(\Lambda\).

\subsection{Sharpness of the Lieb-Robinson bound}
\label{sec:optimality-of-the-LRB}

The bounds in \cref{cor:LRB-commuting-interaction-norm} are sharp and we now construct an example `protocol' which attains the bound.
This easily implies sharpness of \cref{thm:LRB-long-range-alpha-ge-nu} as well, as we explain afterwards.
The main difference to protocols for non-commuting interactions as in~\cite{KS2020,TCE2020} is that due to commutativity only terms in the Hamiltonian directly connecting the supports \(X\) and \(Y\) have an influence.

\subsubsection{Sharpness of Corollary~\ref{cor:LRB-commuting-interaction-norm}}
\label{sec:sharpness-LRB-commutator-version}

Let \(\Lambda\) be any finite graph and consider two-dimensional on-site Hilbert spaces \(\Hi_x = \C^2\), and denote with \(\pauli{\#}_x\) the Pauli~\(\#\) matrix on \(\Hi_x\).
We fix \(X\), \(Y\subset \Lambda\).
Then, for some \(C>0\) let
\begin{equation*}
    \Psi(Z)
    =
    \begin{cases*}
        F_\alpha\pdist{x,y} \, C \, U_{x,y}
         & \text{if \(Z=\List{x,y}\) for some \(x\in X\) and \(y\in Y\)}
        \\
        0
         & \text{else,}
    \end{cases*}
\end{equation*}
where
\begin{equation*}
    U_{x,y}
    =
    \pauliZ_x
    \pauliZ_y
    =
    \proj{\ua\ua}_{xy}
    - \proj{\ua\da}_{xy}
    - \proj{\da\ua}_{xy}
    + \proj{\da\da}_{xy}
\end{equation*}
has norm \(\norm{U_{x,y}} = 1\), such that \(\norm{\Psi}_{F_\alpha} = C\).
Clearly,
\begin{equation*}
    \e^{\I t U_{x,y}}
    =
    \e^{+\I t} \, \proj{\ua\ua}_{xy}
    + \e^{-\I t} \, \proj{\ua\da}_{xy}
    + \e^{-\I t} \, \proj{\da\ua}_{xy}
    + \e^{+\I t} \, \proj{\da\da}_{xy}
    ,
\end{equation*}
so that the dynamics generated by \(U_{x,y}\) add a phase \(\e^{\pm\I t}\) in the computational basis.

In the following, we abbreviate \(
    \ket{a}\ket{b}\ket{c}
    =
    \ket{a}_X \otimes \ket{b}_Y \otimes \ket{c}_{\Lambda\setminus(X\cup Y)}
\), \(\ketU = \ket{\uparrow \dots \uparrow}\) and \(\ketD = \ket{\downarrow \dots \downarrow}\).
Let \(\psi\) be any state on \(\Lambda\setminus(X\cup Y)\).
Since \(U_{x,y}\) mutually commute, we have
\begin{equation*}
    \begin{aligned}
        \e^{\I t H_\Lambda}
        \, \ketU \ketU \ket{\psi}
        &=
        \prod_{x\in X}
        \, \prod_{x\in X}
        \, \e^{\I t C F_\alpha\pdist{x,y} U_{x,y}}
        \, \ketU \ketU \ket{\psi}
        \\&=
        \e^{
            \I t C \sum_{x\in X} \sum_{y\in Y} F_\alpha\pdist{x,y}
        }
        \, \ketU \ketU \ket{\psi}
        \\&=:
        \e^{\I t c} \, \ketU \ketU \ket{\psi}
    \end{aligned}
\end{equation*}
and similarly
\begin{align*}
    \e^{\I t H_\Lambda}
    \ketU \ketD \ket{\psi}
    &=
    \e^{-\I t c}
    \, \ketU \ketD \ket{\psi}
    ,\\
    \e^{\I t H_\Lambda}
    \ketD \ketU \ket{\psi}
    &=
    \e^{-\I t c}
    \, \ketD \ketU \ket{\psi}
    \quad\text{and}\\
    \e^{\I t H_\Lambda}
    \ketD \ketD \ket{\psi}
    &=
    \e^{\I t c}
    \, \ketD \ketD \ket{\psi}
    ,
\end{align*}
where \(c = C \sum_{x\in X} \sum_{y\in Y} F_\alpha\pdist{x,y}\).
Similarly to the protocol in \cref{remark:optimality}, we choose \(A = \prod_{x\in X} \pauliX_x\) such that \(A \, \ketD \ket{\phi} \ket{\psi} = \ketU \ket{\phi} \ket{\psi}\) for all states \(\phi\) on \(Y\).
Moreover, we choose \(B = \ketU \braD_Y + \ketD \braU_Y\), which acts as identity on \(\ketU + \ketD\).
Hence,
\begin{align*}
    \commutator[\big]{
        \tau_t(A), B
    }
    \ketD \paren[\big]{\ketD + \ketU} \ket{\psi}
    &=
    (1-B)
    \, \e^{\I t H} \, A \, \e^{-\I t H}
    \, \ketD \paren[\big]{\ketD + \ketU} \ket{\psi}
    \\&=
    (1-B)
    \, \ketU \paren[\big]{\e^{-2\I t c} \, \ketD + \e^{2\I t c} \, \ketU} \ket{\psi}
    \\&=
    \, \ketU \paren[\Big]{
        \paren[\big]{\e^{-2\I t c} - \e^{2\I t c}} \, \ketD
        + \paren[\big]{\e^{2\I t c} - \e^{-2\I t c}} \, \ketU
    } \ket{\psi}
    ,
\end{align*}
and thus
\begin{equation*}
    \norm[\big]{
        \commutator[\big]{
            \tau_t(A), B
        }
    }
    \geq
    \abs[\big]{
        \e^{2\I t c} - \e^{-2\I t c}
    }
    =
    2 \, \sin\paren[\Big]{
        2 \, t \, \norm{\Psi}_{F_\alpha} \, \sum_{x\in X} \sum_{y\in Y} F_\alpha\pdist{x,y}
    }
    .
\end{equation*}
Using that \(\sin(x) \approx x\) for small \(x\), we see that the LRB obtained in \cref{item-cor:LRB-commuting-interaction-norm:commutator} is sharp whenever its right-hand side is indeed small.

As discussed before, the operator localization version of the LRB implies the commutator version of the LRB\@.
More precisely, \cref{item-cor:LRB-commuting-interaction-norm:operator-localization} implies
\begin{equation*}
    \norm[\big]{
        \commutator{ \tau^\Lambda_t(A), B }
    }
    \leq
    8 \, \norm{\Psi}_F \, \norm{A} \, \norm{B} \, \abs{t}
    \, \sum_{x\in X} \sum_{y\in \Lambda\setminus X_r} F\pd{x,y}
\end{equation*}
for all \(A\in \alg_X\) and \(B\in \alg_{\Lambda\setminus X_r}\).
Hence, the above argument implies that \cref{item-cor:LRB-commuting-interaction-norm:operator-localization} is also optimal up to a factor~\(2\).

\subsubsection{Sharpness of Theorem~\ref{thm:LRB-long-range-alpha-ge-nu}}
To understand that the scaling in the long-range LRBs given in~\eqref{eq:long-range-LR-bound-alpha-ge-nu} and~\eqref{eq:long-range-LR-bound-alpha-ge-nu-operator-localization} is optimal, it is enough to understand that the bound
\begin{equation*}
    \sum_{x\in X} \sum_{y\in Y} F_\alpha\pd{x,y}
    \leq
    \abs{X}
    \, \frac{C}{4} \, F_{\alpha-D}\pd{X,Y}
    \quad\text{for all disjoint \(X\), \(Y\subset \Gamma\)}
    ,
\end{equation*}
which was used to obtain the results from \cref{cor:LRB-commuting-interaction-norm}, is also optimal.
And indeed, for \(X=\List{0}\) and \(Y = \Set{y\in \Gamma \given \dist{x,y} \geq R}\) on \(\Gamma=\Z^D\), it is straightforward to check that
\begin{equation*}
    \sum_{y\in Y} F_\alpha\pdist{0,y}
    =
    \sum_{n=R}^\infty
    \sumstack[r]{y\in \Gamma\suchthat\\\dist{0,y}=n}
    \, F_\alpha(n)
    \geq
    \sum_{n=R}^\infty
    \, \abs{Q_n}
    \, F_\alpha(n)
    =
    \tfrac{1}{(D-1)!}
    \, \sum_{n=R}^\infty
    F_{\alpha-(D-1)}(n)
    ,
\end{equation*}
where \(
    Q_n
    =
    \Set[\big]{
        z\in \R^{D-1}
        \given
        \sum_{k=1}^{D-1} z_k \leq n+1,
        \forall k: z_k \geq 0
    }
\).
Moreover, this sequence only converges for \(\alpha>D\) and then can be lower bounded by
\begin{equation*}
    \sum_{n=R}^\infty
    F_{\alpha-(D-1)}(n)
    \geq
    \int_R^\infty \frac{1}{(r+1)^{\alpha-(D-1)}} \diff r
    =
    \frac{1}{\alpha-D}
    \, F_{\alpha-D}(R)
    .
\end{equation*}
Hence, the obtained decay in \cref{thm:LRB-long-range-alpha-ge-nu} is optimal.

\section{Exemplary applications}
\label{sec:applications}

In recent years, LRBs have been used for a vast range of applications.
Here, we improve some of the results from the literature for commuting interactions.
In particular, we discuss decay of correlations and the \enquote*{local perturbations perturb locally} principle in gapped ground states of long-range Hamiltonians in \cref{sec:decay-of-correlations,sec:Local-perturbations-perturb-locally-principle}, respectively.
In \cref{sec:LR-bound-interaction-picture}, we show that the LRB for a Hamiltonian, which is a sum of a finite-range, commuting Hamiltonian and a general Hamiltonian, does not depend on the strength of the commuting part.

\subsection{Decay of correlations}
\label{sec:decay-of-correlations}

\Textcite{HK2006} proved that local interactions and a uniform spectral gap imply decay of correlations in the ground state in the thermodynamic limit.
Their argument includes long-range interactions with \(\alpha > D\).
However, they use a LRB that follows from the usual proof and is outperformed by more recent approaches, e.g.~\cite{MKN2016,TGB2021}, even for non-commuting Hamiltonians.

For commuting Hamiltonians, the improved LRBs also result in improved decay of correlations as the following \namecref{thm:spectral-gap-implies-decay-of-correlations} shows.
The necessary adjustments to the original proof are discussed in \cref{sec:proof-spectral-gap-implies-decay-of-correlations}.

\begin{theorem}[Spectral gap implies decay of correlations]
    \label{thm:spectral-gap-implies-decay-of-correlations}
    Let \(\Psi\) be a commuting interaction on a finite graph~\(\Lambda\) and assume that the spectrum of \(H\) has a gap \(g>0\) above the ground state energy \(E_0\) in the sense that
    \begin{equation*}
        E_0 \in \sigma(H)
        ,\quad
        E_0
        <
        \sigma(H)\setminus\List{E_0}
        \quadtext{and}
        \dist[\big]{\List{E_0} , \sigma(H)\setminus\List{E_0}}
        \geq
        g
        .
    \end{equation*}
    Let \(P_0\) be the projection onto the ground state sector and \(\rho_0\) be any ground state, i.e.~\(\rho_0 = P_0 \, \rho_0 \, P_0\).

    If \(\Psi\) is polynomially decaying with \(\norm{\Psi}_{F_\alpha}<\infty\) for some \(\alpha>0\) where \(F_\alpha(r) := (1+r)^{-\alpha}\), then, for all disjoint \(X\) and \(Y\subset \Lambda\) and \(A\in \alg_X\) and \(B\in \alg_Y\), it holds that
    \begin{equation}
        \label{eq:thm-decay-of-correlations-long-range}
        \begin{aligned}
            \Alignindent
            \abs[\Big]{
                \trace[\big]{
                    \rho_0
                    \, A
                    \, B
                }
                - \tfrac{1}{2} \paren[\Big]{
                    \trace[\big]{
                        \rho_0
                        \, A
                        \, P_0
                        \, B
                    }
                    +
                    \trace[\big]{
                        \rho_0
                        \, B
                        \, P_0
                        \, A
                    }
                }
            }
            \\&\leq
            8
            \, \norm{A} \, \norm{B}
            \, \abs{X} \, \abs{Y}
            \, \norm{\rho_0}_1
            \, \paren[\bigg]{
                \sqrt{\frac{\alpha}{\pi}}
                \,
                \frac
                    {\norm{\Psi}_{F_\alpha}}
                    {g}
                +
                1
            }
            \, \ln(1+r)
            \, F_\alpha(r)
            ,
        \end{aligned}
    \end{equation}
    where \(r=\dist{X,Y}\).

    If \(\Psi\) is a short-range interaction satisfying \(\norm{\Psi}_{F_{b,p}}<\infty\) for some \(b>\tilde{b}>0\) and \(p\in \intervaloc{0,1}\) where \(F_{b,p}(r) := \e^{-b r^p}\), then, for all disjoint \(X\) and \(Y\subset \Lambda\) and \(A\in \alg_X\) and \(B\in \alg_Y\), it holds that
    \begin{equation}
        \label{eq:thm-decay-of-correlations-short-range}
        \begin{aligned}
            \Alignindent
            \abs[\Big]{
                \trace[\big]{
                    \rho_0
                    \, A
                    \, B
                }
                - \tfrac{1}{2} \paren[\Big]{
                    \trace[\big]{
                        \rho_0
                        \, A
                        \, P_0
                        \, B
                    }
                    +
                    \trace[\big]{
                        \rho_0
                        \, B
                        \, P_0
                        \, A
                    }
                }
            }
            \\&\leq
            \norm{A} \, \norm{B} \, \min\List[\big]{\abs{X},\abs{Y}} \, \norm{\rho_0}_1
            \, \paren[\bigg]{
                \frac{C \, \norm{\Psi}_{b,p}}{g}
                +1
            }
            \, F_{\tilde{b},p}(r)
            ,
        \end{aligned}
    \end{equation}
    where \(r=\dist{X,Y}\) and \(C\) is an explicit constant given in~\eqref{eq:constant-in-decay-of-correlations-short-range} and only depends on \(b\) and \(\tilde{b}\).
\end{theorem}

The results concerning decay of correlations from~\cite{HK2006} have been improved for general (non-commuting) interactions in~\cite{WH2022} using analogous methods to those used in~\cite{WH2022} to prove LPPL, which we will discuss in the next section.
In the case of commuting interactions, the improved method qualitatively yields the same results, and we decided to use the previous method by~\cite{HK2006} to emphasize the broad applicability of improved LRBs for commuting Hamiltonians.

Not that the first lines in~\eqref{eq:thm-decay-of-correlations-long-range} and~\eqref{eq:thm-decay-of-correlations-short-range} simplify to the more common
\begin{equation*}
    \abs[\big]{
        \trace{
            \rho_0
            \, A
            \, B
        }
        -
        \trace{
            \rho_0
            \, A
        }
        \,
        \trace{
            \rho_0
            \, B
        }
    }
\end{equation*}
in the case \(\rho_0=P_0\).

We only give a finite volume version of the statement in~\cite{HK2006} here.
For \(\alpha > D\) (or short-range interactions) there exists a thermodynamic, see~\eqref{eq:definition-thermodynamic-limit}.
Then, a statement similar to \cref{thm:spectral-gap-implies-decay-of-correlations} also holds in the thermodynamic limit, where it is enough that a gapped spectral patch converges to a unique ground state energy as in~\cite[(2.12)]{HK2006}.

Note, that the decay exponent in~\eqref{eq:thm-decay-of-correlations-long-range} does not depend on the gap and equals the decay of the interaction.
Thus, it is better than the trivial bound \(2 \, \norm{A} \, \norm{B} \, \norm{\rho_0}_1\) for~\(r \gtrsim g^{-1/\alpha}\).
Hence, this result is a qualitative improvement over the original one in~\cite{HK2006} and also the more recent~\cite{WH2022}.
In both references, the decay exponent in the bound scales like \(\alpha \, g\) for small~\(g\), meaning that the bound becomes non-trivial only for \(\ln r \gtrsim g^{-1}\).

Moreover, this qualitative improvement is also apparent in the result for short-range interactions: The previous results in~\cite{HK2006,WH2022} both have a correlation length \(\xi \sim 1/g\) for small~\(g\), i.e.\ the bounds they obtain scale with \(\e^{-\dist{X,Y}/\xi}\).
In~\cite{GH2016} an improved scaling \(\xi\sim 1/\sqrt{g}\) was obtained for frustration-free Hamiltonians.
In contrast, here we prove that the correlation length is independent of the gap, \(\xi \sim 1\) if the underlying interaction is commuting.

\begin{remark}
    For \(\alpha>D\), one can also use~\eqref{eq:long-range-LR-bound-alpha-ge-nu} and obtain the same statement with the bound from~\eqref{eq:thm-decay-of-correlations-long-range} replaced by
    \begin{equation*}
        8 \, C \, \norm{A} \, \norm{B} \, \min\List{\abs{X}, \abs{Y}} \, \norm{\rho_0}_1
        \, \norm{\Psi}_{F_\alpha}
        \, \paren[\bigg]{
            \sqrt{\frac{\alpha}{\pi}}
            \,
            \frac
                {\norm{\Psi}_{F_\alpha}}
                {g}
            +
            1
        }
        \, \ln(1+r) \, F_{\alpha-D}(r)
        ,
    \end{equation*}
    where \(C\) is the constant from~\eqref{eq:long-range-LR-bound-alpha-ge-nu}.
\end{remark}

\subsection{Local perturbations perturb locally principle}
\label{sec:Local-perturbations-perturb-locally-principle}

\Textcite{WH2022} recently proved a version of the local perturbations perturb locally (LPPL) principle for long-range systems with a new technique.
They avoid the spectral flow~\cite{HW2005}, which was used by previous results for short-range Hamiltonians~\cite{BMNS2012}.
The general idea is that the ground state of a gapped system only changes locally around a small perturbation, in the sense that expectations values of local observables away from the perturbation do not change.

With the improved LRBs we can also improve these results for gapped commuting Hamiltonians.

\begin{theorem}[LPPL for gapped ground states]
    \label{thm:LPPL-for-gapped-ground-states}
    Let \(\Psi\) be a commuting interaction on a finite graph~\(\Lambda\) and let \(V\in \alg_X\) with \(X\in \Lambda\) be some perturbation.
    Moreover, assume that \(H+\lambda\,V\) has a unique ground state \(\rho_\lambda\) and a gap of size at least \(g>0\) above the ground state for all \(\lambda\in \intervalcc{0,1}\).

    If \(\Psi\) is polynomially decaying with \(\norm{\Psi}_{F_\alpha}<\infty\) for some \(\alpha>0\) where \(F_\alpha(r) := (1+r)^{-\alpha}\), then, for all \(Y\subset \Lambda\) and \(B\in \alg_Y\), it holds that
    \begin{equation}
        \label{eq:thm-LPPL-gapped-ground-states-long-range}
        \abs[\Big]{
            \trace[\big]{\rho_0 \, B}
            - \trace[\big]{\rho_1 \, B}
        }
        \leq
        32 \, \norm{\Psi}_{F_\alpha} \, \norm{B} \, \paren[\big]{\norm{V}+\norm{V}^2}
        \, \abs{X} \, \abs{Y}
        \, F_\alpha(r) \, \frac{g+2}{g^3}
        ,
    \end{equation}
    where \(r=\dist{X,Y}\).

    If \(\Psi\) is a short-range interaction satisfying \(\norm{\Psi}_{F_{b,p}}<\infty\) for some \(b>b'>0\) and \(p\in \intervaloc{0,1}\) where \(F_{b,p}(r) := \e^{-b r^p}\), then, for all \(Y\subset \Lambda\) and \(B\in \alg_Y\), it holds that
    \begin{equation*}
        \label{eq:thm-LPPL-gapped-ground-states-short-range}
        \abs[\Big]{
            \trace[\big]{\rho_0 \, B}
            - \trace[\big]{\rho_1 \, B}
        }
        \leq
        8 \, C_{b,b'}
        \, \norm{\Psi}_{b,p}
        \, \norm{B}
        \, \paren[\big]{\norm{V}+\norm{V}^2}
        \, \min\List[\big]{\abs{X}, \abs{Y}}
        \, F_{b',p}(r)
        \, \frac{g+2}{g^3}
        ,
    \end{equation*}
    where \(r=\dist{X,Y}\), and \(C_{b,b'}\) is the constant from~\eqref{eq:short-range-LR-bound}.
\end{theorem}

As in the result on decay of correlations, the decay exponent in~\eqref{eq:thm-LPPL-gapped-ground-states-long-range} does not depend on the gap and equals the decay of the interaction.
It is better than the trivial bound~\(2 \, \norm{B}\) for~\(r \gtrsim g^{-3/\alpha}\).
Thus, it is a qualitative improvement compared to the bound in~\cite{WH2022} for general interactions, where the exponent scales like \(\alpha\,g\) for small \(g\), which make the bound better than the trivial one for \(\ln r \gtrsim g^{-1}\).
The same is also apparent in the short-range setting, where we obtain a decay \(\e^{-b' \dist{X,Y}^p}\) for any \(b'<b\) uniformly in \(g>0\) for commuting interactions, while~\cite{WH2022} only proves a stretched-exponential decay with \(b' \sim b \, g\) for small~\(g\) for general interactions.

To obtain the statement, we need improved LRBs also for the perturbed Hamiltonian \(H+V\).
And since the statement is trivial, if also the perturbation~\(V\) commutes with all Hamiltonian terms, we cannot just rely on the results from \cref{sec:LR-bound-long-range}.
Instead, we use a previous result on LRBs for perturbed Hamiltonians~\cite[Lemma~33]{CMTW2023}, details are discussed in \cref{sec:proof-LPPL-for-gapped-ground-states}.

\begin{remark}
    For \(\alpha>D\), one can also use~\eqref{eq:long-range-LR-bound-alpha-ge-nu} and obtain the same statement with the bound from~\eqref{eq:thm-LPPL-gapped-ground-states-long-range} replaced by
    \begin{equation*}
        32 \, C \, \norm{\Psi}_{F_\alpha} \, \norm{B} \, \paren[\big]{\norm{V}+\norm{V}^2}
        \, \min\List{\abs{X}, \abs{Y}}
        \, F_{\alpha-D}(r) \, \frac{g+2}{g^3}
        ,
    \end{equation*}
    where \(C\) is the constant from~\eqref{eq:long-range-LR-bound-alpha-ge-nu}.
\end{remark}

\subsection{General Hamiltonians with a commuting part}
\label{sec:LR-bound-interaction-picture}

In contrast to the rest of the paper, we want to investigate non-commuting Hamiltonians which only have a commuting part in this section.
More specifically, we assume to be given two interactions \(\Phi\) and \(\Psi\), where \(\Psi\) is commuting as before, \(\Phi\) might not be commuting and in particular \(\Psi\) and \(\Phi\) are not required to commute.
For simplicity, we assume that the commuting interaction \(\Psi\) is of finite range \(R>0\), meaning that \(\Psi(Z) = 0\) unless \(\diam{Z} > R\).
The arguments work similarly for short- or long-range interactions, as long as the decay of the commuting part \(\Psi\) is better than the one of the non-commuting part \(\Phi\).
Moreover, we allow \(\Phi\) to be time-dependent.
We write \(\Phi(t)\) for the interaction at time \(t\), \(\Phi(Z,t):=\Phi(t)(Z)\) the interaction term at \(Z\subset \Lambda\) at time \(t\), and
\begin{equation*}
    \norm{\Phi}_F := \sup_{t\in I} \norm{\Phi(t)}_F
\end{equation*}
for the norm of a time-dependent interaction \(\Phi\).
The time interval \(I\subset \R\) of interest will be clear from the context.

In this setting we obtain the following LRB for the evolution of the full system.

\begin{theorem}
    Let \(\Lambda\) be a finite graph and \(F\colon \R_{\geq 0} \to \R_{>0}\) a decaying function.
    Assume that we have a general Lieb-Robinson bound given in terms of a function \(\zeta\), such that for any interaction \(\Phi\) with \(\norm{\Phi}_F < \infty\), it holds that for all \(X\),~\(Y\subset \Lambda\), \(A\in \alg_X\) and \(B\in \alg_Y\)
    \begin{equation*}
        \norm[\big]{
            \commutator{ \tau_{t,s}(A), B }
        }
        \leq
        \zeta\paren[\Big]{
            \norm{\Phi}_F,
            \norm{A},
            \norm{B},
            \abs{X},
            \abs{Y},
            \dist{A,B},
            \abs{t-s}
        }
        .
    \end{equation*}

    Furthermore, let \(\Psi\) be a commuting interaction of range \(R>0\) and let \(\Phi\) be a general, time-dependent interaction on \(\Lambda\) satisfying \(\norm{\Phi}_{F(\mathop{\cdot}+2R)} < \infty\).
    Then the dynamics \(\tau_{t,s}\) for the sum \(\Phi + \Psi\) satisfy the following modified Lieb-Robinson bound:
    For all \(X\),~\(Y\subset \Lambda\), \(A\in \alg_X\) and \(B\in \alg_Y\) it holds that
    \begin{equation*}
        \norm[\big]{
            \commutator{ \tau_{t,s}(A), B }
        }
        \leq
        \zeta\paren[\Big]{
            c^2 \, \norm{\Phi}_{F(\mathop{\cdot}+2R)},
            \norm{A},
            \norm{B},
            c\,\abs{X},
            c\,\abs{Y},
            \dist{A,B}-2R,
            \abs{t-s}
        }
        ,
    \end{equation*}
    with \(c=(1+\kappa) \, R^{D}\).
\end{theorem}

In particular, the Lieb-Robinson velocity and the bound do not depend on the norm of the commuting part \(\Psi\).
As long as the evolution exists it could even be unbounded if one works in a setting with infinite dimensional local Hilbert spaces.

As an example we provide the following result for exponentially decaying interactions, where we use the LRB from~\cite{NSY2019}.

\begin{corollary}
    Let \(\Lambda\) be a finite graph, \(\Psi\) a commuting interaction of range \(R>0\) and \(\Phi\) an exponentially decaying interaction with \(\norm{\Phi}_{F(\mathop{\cdot}+2R)} < \infty\) where \(F(r) := \e^{-br}/(1+r)^{D+1}\).
    Then the full dynamics \(\tau_{t,s}\) generated by the interaction \(\Psi + \Phi\) satisfies the following Lieb-Robinson bound:
    For all \(X\),~\(Y\subset \Lambda\), \(A\in \alg_X\) and \(B\in \alg_Y\), it holds that
    \begin{equation*}
        \norm[\big]{
            \commutator{ \tau_{t,s}(A), B }
        }
        \leq
        C
        \, \norm{A} \, \norm{B}
        \, \min\List[\big]{\abs{X},\abs{Y}}
        \, \e^{b \, (v_b \abs{t-s} - \dist{X,Y})}
        ,
    \end{equation*}
    where \(C = 2 \, \norm{F} \, C_F^{-1} \, \e^{2bR}\) and \(
        v_b
        =
        2 \, C_F \, (1+\kappa)^2 \, R^{2 D} \, \norm{\Phi}_{F(\mathop{\cdot}+2R)}/b
    \), and \(\norm{F}\) and \(C_F\) are constants depending on \(F\) that are defined in~\cite[equations~(3.8) and~(3.9)]{NSY2019}.
\end{corollary}

Again, we highlight that the velocity does not depend on the strength of \(\Psi\), but only on its range.

As an easy example, consider the XXZ spin chain on the Hilbert space \(\bigotimes_{x=1}^L\C^{2S+1}\)
\begin{equation*}
    H_{XXZ}
    :=
    -\sum_{x=1}^{L-1} \paren*{S_{x}^1S_{x+1}^1+S_{x}^2S_{x+1}^2+\Delta S_{x}^3
    S_{x+1}^3}
\end{equation*}
with anisotropy parameter \(\Delta>1\) and each \(\vec{S}_x\) an irreducible spin-\(S\) representation of \(\mathfrak{su}(2)\).
Applying the above result to the commuting interactions \(\Psi(\{x,x+1\}) := \Delta S_{x}^3 \, S_{x+1}^3\), we obtain a LRB that is independent of \(\Delta\).

\section{Proofs of applications}

\subsection{Proof of Theorem~\ref{thm:spectral-gap-implies-decay-of-correlations}}
\label{sec:proof-spectral-gap-implies-decay-of-correlations}

Since we only consider finite volumes with a single ground state energy, the proof simplifies slightly compared to~\cite{HK2006}.
To get the constants right, we repeat large parts of the proof but leave out some of the details and the proof of~\cite[Lemma~3.1]{HK2006}.
We first give the proof for~\eqref{eq:thm-decay-of-correlations-long-range} and then comment about the small modifications necessary for~\eqref{eq:thm-decay-of-correlations-short-range}.

First, we observe
\begin{equation}
    \label{eq:proof-DC-first-step}
    \begin{aligned}
        \Alignindent
        \trace[\big]{
            \rho_0
            \, \commutator[\big]{\tau^\Lambda_t(A), B}
        }
        \\&=
        \begin{aligned}[t]
            &
            \trace[\big]{
                \rho_0
                \, \tau^\Lambda_t(A)
                \, (\unit - P_0)
                \, B
            }
            - \trace[\big]{
                \rho_0
                \, B
                \, (\unit - P_0)
                \, \tau^\Lambda_t(A)
            }
            \\&+
            \trace[\big]{
                \rho_0
                \, \tau^\Lambda_t(A)
                \, P_0
                \, B
            }
            - \trace[\big]{
                \rho_0
                \, B
                \, P_0
                \, \tau^\Lambda_t(A)
            }
        \end{aligned}
        \\&=
        \begin{aligned}[t]
            &
            \sum_{n>0}
            \trace[\big]{
                \rho_0
                \, A
                \, P_n
                \, B
            }
            \, \e^{-\I (E_n-E_0) t}
            -
            \sum_{n>0}
            \trace[\big]{
                \rho_0
                \, A
                \, P_n
                \, B
            }
            \, \e^{\I (E_n-E_0) t}
            \\&+
            \trace[\big]{
                \rho_0
                \, A
                \, P_0
                \, B
            }
            -
            \trace[\big]{
                \rho_0
                \, B
                \, P_0
                \, A
            }
            ,
        \end{aligned}
    \end{aligned}
\end{equation}
where \(P_n\) are the spectral projections for eigenvalues \(E_n\), and we used
\begin{equation*}
    \e^{\I H t} \, \rho_0 = \e^{\I E_0 t} \, \rho_0
    \qquadtext{and}
    \e^{\I H t} \, P_n = \e^{\I E_n t} \, P_n
    \quadtext{for all}
    n
    .
\end{equation*}
As in~\cite{HK2006} we now apply
\begin{equation*}
    \calF(\cdot)
    :=
    \lim_{T \to \infty} \, \lim_{\epsi \to 0} \, \frac{\I}{2\pi} \, \int_{-T}^T \frac{(\cdot) \, \e^{-\beta t^2}}{t+\I\epsi} \diff t
\end{equation*}
on both sides of~\eqref{eq:proof-DC-first-step}.
Instead of~\cite[(3.24)]{HK2006} we use the LRB from~\eqref{eq:long-range-LR-bound-alpha-ge-0}.
Thus, for the left-hand side we bound
\begin{equation}
    \label{eq:proof-decay-of-correlations-bound-F(LHS)}
    \abs[\bigg]{
        \frac{\I}{2\pi}
        \int_{-T}^T \frac
            {\trace[\big]{\rho_0 \, \commutator[\big]{\tau^\Lambda_t(A), B}}}
            {t+\I \epsi}
        \, \e^{-\beta t^2}
        \diff t
    }
    \leq
    \frac{C_1}{2\pi} \, F_\alpha(r) \, \int_{-\infty}^\infty \e^{-\beta t^2} \diff t
    \leq
    \frac{C_1}{2\,\sqrt{\pi\beta}} \, F_\alpha(r)
    ,
\end{equation}
with \(
    C_1
    =
    4 \, \norm{\Psi}_{F_\alpha} \, \norm{A} \, \norm{B} \, \abs{X} \, \abs{Y} \, \norm{\rho_0}_1
\), which also persists the two limits.
For the two time-independent terms on the right-hand side, we calculate
\begin{equation*}
    \lim_{T \to \infty} \, \lim_{\epsi \to 0} \, \frac{\I}{2\pi} \, \int_{-T}^T \frac{\e^{-\beta t^2}}{t+\I\epsi} \diff t
    =
    \frac{1}{2}
    .
\end{equation*}
And by~\cite[Lemma~3.1]{HK2006}, which – after inspection of the constants – states
\begin{align}
    \label{eq:proof-DC-Lemma-HK2006-E-geq-0}
    \abs[\Big]{\calF\paren[\big]{\e^{-\I E t}} - 1}
    &\leq
    \frac{1}{2} \, \e^{-E^2/(4\beta)}
    \quadtext{for}
    E>0
\shortintertext{and}
    \label{eq:proof-DC-Lemma-HK2006-E-leq-0}
    \abs[\Big]{\calF\paren[\big]{\e^{-\I E t}}}
    &\leq
    \frac{1}{2} \, \e^{-E^2/(4\beta)}
    \quadtext{for}
    E<0
    ,
\end{align}
for all \(\beta>0\), the remaining two terms on the right-hand side satisfy
\begin{align*}
    \abs[\bigg]{
        \calF\paren[\bigg]{
            \sum_{n>0}
            \trace[\big]{
                \rho_0
                \, A
                \, P_n
                \, B
            }
            \, \e^{-\I (E_n-E_0) t}
        }
        -
        \trace[\big]{
            \rho_0
            \, A
            \, \paren[\big]{\unit - P_0}
            \, B
        }
    }
    &\leq
    C_2 \, \e^{-g^2/(4\beta)}
\shortintertext{and}
    \abs[\bigg]{
        \calF\paren[\Big]{
            \sum_{n>0}
            \trace[\big]{
                \rho_0
                \, A
                \, P_n
                \, B
            }
            \, \e^{\I (E_n-E_0) t}
        }
    }
    &\leq
    C_2 \, \e^{-g^2/(4\beta)}
    ,
\end{align*}
with \(C_2 = 1/2 \, \norm{A} \, \norm{B} \, \norm{\rho_0}_1\).
To prove these bounds, we first use triangle inequality and then apply~\eqref{eq:proof-DC-Lemma-HK2006-E-geq-0} and~\eqref{eq:proof-DC-Lemma-HK2006-E-leq-0}, respectively, together with the upper bound \(E_n-E_0 \geq g\).
Then it is left to bound
\begin{equation*}
    \sum_{n > 0}
    \, \abs[\big]{
        \trace[\big]{
            \rho_0
            \, A
            \, P_n
            \, B
        }
    }
    \leq
    \sum_{m} \abs[\big]{
        \braket{\phi_m, B \, \rho_0 \, A \, \phi_m}
    }
    \leq
    \norm{B \, \rho_0 \, A}_1
    \leq
    \norm{A} \, \norm{B} \, \norm{\rho_0}_1
    ,
\end{equation*}
where \(\List{\phi_m}\) is any orthonormal energy eigenbasis, i.e.\ for every \(n\) there exists an index set \(M_n\) such that \(P_n = \sum_{m\in M_n} \ket{\phi_m}\bra{\phi_m}\).
In the second step we used that for every operator \(T\) and every ONB \(\phi_m\) it holds that
\begin{equation*}
    \sum_{m}
    \abs[\big]{
        \braket{\phi_m, T \, \phi_m}
    }
    =
    \sum_{m}
    \, \braket{\phi_m, T \, D \, \phi_m}
    =
    \trace{T \, D}
    \leq
    \norm{T}_1 \, \norm{D}
    ,
\end{equation*}
where we choose \(\theta_m\in \C\) with \(\abs{\theta_m}=1\) such that
\(
    \abs[\big]{
        \braket{\phi_m, T \, \phi_m}
    }
    =
    \theta_m \, \braket{\phi_m, T \, \phi_m}
\)
and \(D = \sum_m \theta_m \, \proj{\phi_m}\), which satisfies \(\norm{D}=1\).

In total, we obtain
\begin{align*}
    \Alignindent
    \abs[\Big]{
        \trace[\big]{
            \rho_0
            \, A
            \, B
        }
        - \tfrac{1}{2} \paren[\Big]{
            \trace[\big]{
                \rho_0
                \, A
                \, P_0
                \, B
            }
            +
            \trace[\big]{
                \rho_0
                \, B
                \, P_0
                \, A
            }
        }
    }
    \\&\leq
    \frac{C_1}{2\,\sqrt{\pi\beta}} \, F_\alpha(r)
    + 2 \, C_2 \, \e^{-g^2/(4\beta)}
    \\&\leq
    \norm{A} \, \norm{B} \, \norm{\rho_0}_1
    \, \paren[\bigg]{
        \frac
            {4 \, \norm{\Psi}_{F_\alpha} \, \sqrt{\alpha}}
            {\sqrt{\pi} \, g}
        \, \abs{X} \, \abs{Y}
        \, \ln(1+r)
        +
        1
    }
    \, F_\alpha(r)
    \\&\leq
    8
    \, \norm{A} \, \norm{B}
    \, \abs{X} \, \abs{Y}
    \, \norm{\rho_0}_1
    \, \paren[\bigg]{
        \sqrt{\frac{\alpha}{\pi}}
        \,
        \frac
            {\norm{\Psi}_{F_\alpha}}
            {g}
        +
        1
    }
    \, \ln(1+r)
    \, F_\alpha(r)
    ,
\end{align*}
where we chose
\begin{equation*}
    \beta = \dfrac{g^2}{4 \, \alpha \, \ln(1+r)}
\end{equation*}
in the second step.
For the last step we use \(\ln(1+r)/\ln 2 \geq 1\) to move the logarithm out of the parenthesis, followed by bounding \(\ln 2 \leq 1\) and \(1/\ln 2 \leq 2\).

To obtain the result for short-range interactions, we only have to modify~\eqref{eq:proof-decay-of-correlations-bound-F(LHS)} and the final bound.
We choose \(b' = (b+\tilde{b})/2 \in \intervaloo{\tilde{b},b}\) and use the LRB from \cref{thm:LR-bound-short-range} for \(0<b'<b\).
Due to the similar structure of the LRBs, we obtain
\begin{equation*}
    \abs[\Big]{
        \calF\paren[\Big]{
            \trace[\big]{
                \rho_0
                \, \commutator[\big]{\tau^\Lambda_t(A), B}
            }
        }
    }
    \leq
    \frac{C_1}{2\,\sqrt{\pi\beta}} \, F_{b',p}(r)
\end{equation*}
with
\(
    C_1
    =
    C_{b,b'}
    \, \norm{\Psi}_{F_{b,p}}
    \, \norm{A} \, \norm{B}
    \, \min\List[\big]{\abs{X}, \abs{Y}}
    \, \norm{\rho_0}_1
\)
and \(C_{b,b'}\) the constant from \cref{thm:LR-bound-short-range}.
In total, we then obtain
\begin{align*}
    \Alignindent
    \abs[\Big]{
        \trace[\big]{
            \rho_0
            \, A
            \, B
        }
        - \tfrac{1}{2} \paren[\Big]{
            \trace[\big]{
                \rho_0
                \, A
                \, P_0
                \, B
            }
            +
            \trace[\big]{
                \rho_0
                \, B
                \, P_0
                \, A
            }
        }
    }
    \\&\leq
    \frac{C_1}{2\,\sqrt{\pi\beta}} \, F_{b',p}(r)
    + 2 \, C_2 \, \e^{-g^2/(4\beta)}
    \\&\leq
    \norm{A} \, \norm{B}
    \, \min\List[\big]{\abs{X}, \abs{Y}}
    \, \norm{\rho_0}_1
    \paren[\bigg]{
        \sqrt{\frac{\tilde{b}}{\pi}}
        \, \frac
        {C_{b,b'} \, \norm{\Psi}_{F_{b,p}}}
        {g}
        \, r^{p/2}
        \, F_{b',p}(r)
        +
        F_{\tilde{b},p}(r)
    }
\end{align*}
after choosing \(\beta=\dfrac{g^2}{4\,\tilde{b}\,r^p}\).
With~\cite[Lemma~7.2.3(b)]{Maier2022} we can bound
\begin{equation*}
    r^{p/2}
    \, F_{b',p}(r)
    \leq
    \paren[\big]{
        2\,\e\,(b'-\tilde{b})
    }^{-1/2}
    \, F_{\tilde{b},p}(r)
    =
    \paren[\big]{
        \e\,(b-\tilde{b})
    }^{-1/2}
    \, F_{\tilde{b},p}(r)
    ,
\end{equation*}
which leads us to the final bound with
\begin{equation}
    \label{eq:constant-in-decay-of-correlations-short-range}
    C
    =
    \paren[\bigg]{
        \frac
            {\tilde{b} \, C_{b,(b+\tilde{b})/2}}
            {\pi \, \e \, (b-\tilde{b})}
    }^{1/2}
    \qquad\text{with \(C_{b,(b+\tilde{b})/2}\) the constant from \cref{thm:LR-bound-short-range}}
    .
\end{equation}

\subsection{Proof of Theorem~\ref{thm:LPPL-for-gapped-ground-states}}
\label{sec:proof-LPPL-for-gapped-ground-states}

As remarked in the main text, we need a LRB for all Hamiltonians \(H+\lambda\,V\) to apply the results from~\cite{WH2022}.
But we only need to estimate commutators, where one of the observables is the perturbation itself.
We begin with the proof of~\eqref{eq:thm-LPPL-gapped-ground-states-long-range}.
Without requiring that \(V\) and \(H\) commute, combining~\cite[Lemma~33]{CMTW2023} and~\eqref{eq:long-range-LR-bound-alpha-ge-0}, we obtain the bound
\begin{equation}
    \label{eq:LR-bound-perturbed-Hamiltonian}
    \norm[\big]{\commutator[\big]{\e^{-\I t (H+\lambda V)} \, B \, \e^{\I t (H+\lambda V)}, V}}
    \leq
    4 \, \norm{\Psi}_{F_\alpha} \, \norm{B} \, \norm{V}
    \, \abs{X} \, \abs{Y}
    \, \abs{t} \, \paren[\big]{1+\norm{V} \, \abs{t}}
    \, F_\alpha\pdist{X,Y}
    .
\end{equation}
In the language of~\cite{WH2022} we have
\begin{equation}
    \label{eq:proof-LPPL-LR-bound-short}
    C(r,t)
    =
    C \, F_\alpha(r) \, \paren[\big]{t+t^2}
    ,
\end{equation}
where \(
    C
    =
    4 \, \norm{\Psi}_{F_\alpha} \, \norm{B} \, \norm{V} \, \abs{X} \, \abs{Y} \, \paren[\big]{\norm{V}+\norm{V}^2}
\).
As in~\cite{WH2022}, we omit the labels \(\lambda\), since all bounds are uniform in those.
Then we calculate, see~\cite[eq.~(11)]{WH2022},
\begin{equation*}
    \bar{\Omega}(r,y)
    :=
    \int_0^\infty
    C(r,t)
    \, \e^{-yt}
    \diff t
    =
    C \, F_\alpha(r) \, \frac{y+2}{y^3}
    .
\end{equation*}
With~\cite[eq.~(13)]{WH2022}, we bound
\begin{align*}
    \ln \abs[\big]{\Omega_{XY}(0)}
    &\leq
    \frac{1}{2\pi}
    \, \int_{0}^{2\pi}
    \ln \bar{\Omega}\paren[\big]{r,\abs{\rho \sin \theta}}
    \diff \theta
    \\&\leq
    \ln\paren[\big]{C \, F_\alpha(r)}
    + \ln\paren[\big]{\rho+2}
    - \ln \rho^3
    - \frac{3}{2\pi} \, \int_{0}^{2\pi} \ln\abs{\sin \theta} \diff \theta
    \\&\leq
    \ln\paren[\bigg]{
        8 \, C \, F_\alpha(r) \, \frac{\rho+2}{\rho^3}
    }
    ,
\end{align*}
for any \(\rho\in \intervaloo{0,g}\).
In the last step, we used \(\int_{0}^{2\pi} \ln\abs{\sin \theta} \diff \theta = -\pi \, \ln 4\).
Now, the remaining arguments from~\cite{WH2022} and the limit \(\rho \to g\) yield the statement.

Since the explicit form of the decay was not used, the proof of~\eqref{eq:thm-LPPL-gapped-ground-states-short-range} is exactly the same, but with
\begin{equation*}
    C
    =
    C_{b,b'} \, \norm{\Psi}_{F_{b,b'}} \, \norm{A} \, \norm{B}
    \, \min\List[\big]{\abs{X}, \abs{Y}} \, \paren[\big]{\norm{V}+\norm{V}^2}
\end{equation*}
and \(C(r,t) = C \, F_{b',p}(r) \, \paren[\big]{t+t^2}\) coming from the LRB given in \cref{thm:LR-bound-short-range}.

\subsection{Lieb-Robinson bounds for Hamiltonians with a commuting part}
\label{sec:proof-lr-bound-interaction-picture}

The idea of the proof is to go to the interaction picture where \(\Psi\) is the unperturbed part and \(\Phi\) is the perturbation.
The range of the interaction picture interaction \(\Phiint\) will only be slightly enlarged by the commuting, finite-range interaction \(\Psi\), while the norm of the individual terms are not changed at all by the unitary transformation.
This allows to apply known LRBs for non-commuting interactions.

The dynamics for the full Hamiltonian \(
    H(t)
    :=
    H^\Psi + H^\Phi(t)
    :=
    \sum_{Z\subset \Lambda} \Psi(Z) + \sum_{Z\subset \Lambda} \Phi(t,Z)
\) is given by the unique solution of
\begin{equation}
    \label{eq:definition-time-evolution-time-dependent-generator}
    -\I \, \odv*{\tau_{t,s}(A)}{t}
    =
    \tau_{t,s}\paren[\big]{
        \commutator[\big]{H(t),A}
    }
    \qquadtext{and}
    \tau_{s,s} = \id
    \qquadtext{for all}
    s,t\in I
    .
\end{equation}
For the interaction picture, we define
\begin{equation}
    \label{eq:def-interaction-picture-Hamiltonian}
    \Phiint(t,Z)
    :=
    \sumstack{X\subset \Lambda\suchthat\\X_R=Z} \tau^\Psi_{t}\paren[\big]{\Phi(t,X)}
    ,
    \qquadtext{such that}
    \Hint(t)
    =
    \tau^\Psi_{t}\paren[\big]{H^\Phi(t)}
    ,
\end{equation}
where \(\Hint(t) := \sum_{Z\subset \Lambda} \Phiint(t,Z)\), \(H^\Phi(s) = \sum_{Z\subset \Lambda} \Phi(t,Z)\), and \(\tau^\Psi\) is the dynamics on \(\Lambda\) generated by \(\Psi\) as given in~\eqref{eq:definition-time-evolution}.
From \cref{sec:LR-bound-finite-range} we know that \(\tau^\Psi_{t}\paren[\big]{\Phi(s,X)}\in \alg_{X_R}\).

We recall that the dynamics \(\tauint_{t,s}\) generated by \(\Phiint\) is the solution of the analogue equation to~\eqref{eq:definition-time-evolution-time-dependent-generator} with generator \(\Hint(t)\).
We conclude that \(\tau_{t,s} = \tau^\Psi_{-s} \circ \tauint_{t,s} \circ \tau^\Psi_t\) by observing that the right-hand side solves~\eqref{eq:definition-time-evolution-time-dependent-generator} since
\begin{align*}
    -\I \, \odv*{}{t}
    \, \tau^\Psi_{-s} \circ \tauint_{t,s} \circ \tau^\Psi_t (A)
    &=
    \tau^\Psi_{-s} \circ \tauint_{t,s} \paren[\Big]{
        \commutator[\big]{\Hint(t) ,\tau^\Psi_t (A)}
    }
    + \tau^\Psi_{-s} \circ \tauint_{t,s} \circ \tau^\Psi_t \paren[\big]{
        \commutator{H^\Psi, A}
    }
    \\&=
    \tau^\Psi_{-s} \circ \tauint_{t,s} \circ \tau^\Psi_t \paren[\big]{
        \commutator{H^\Psi + H^\Phi(t), A}
    }
    .
\end{align*}
To obtain a LRB for the full evolution~\(\tau_{t,s}\) we now look at this decomposition.
First, by~\eqref{eq:def-interaction-picture-Hamiltonian}, for \(x,y\in \Lambda\), we have
\begin{align*}
    \sumstack{Z\subset \Lambda\suchthat\\x,y\in Z}
    \norm[\big]{\Phiint(t,Z)}
    &\leq
    \sumstack{Z\subset \Lambda\suchthat\\x,y\in Z}
    \, \sumstack{X\subset \Lambda\suchthat\\X_R=Z}
    \norm[\big]{\tau^\Psi_{t}\paren[\big]{\Phi(t,X)}}
    \\&\leq
    \, \sumstack{x'\in B_x(R)\\y'\in B_y(R)}
    \, \sumstack{X\subset \Lambda\suchthat\\x',y'\in Z}
    \norm{\Phi(t,X)}
    \\&\leq
    \sumstack{x'\in B_x(R)\\y'\in B_y(R)}
    \, \norm{\Phi}_{F(\mathop{\cdot}+2R)}
    \, F\paren[\big]{\dist{x',y'}+2R}
    .
\end{align*}
Hence, \(
    \norm{\Phiint}_{F}
    \leq
    (1+\kappa)^2 \, R^{2 D} \, \norm{\Phi}_{F(\mathop{\cdot}+2R)}
\), since \(\dist{x',y'} \geq \dist{x,y} - 2R\).
Then, we rewrite
\begin{equation}
    \norm[\big]{
        \commutator{ \tau_{t,s}(A), B }
    }
    =
    \norm[\big]{
        \commutator{ \tauint_{t,s} \circ \tau^\Psi_t (A), \tau^\Psi_s (B) }
    }
\end{equation}
and recall that \(\tau^\Psi_s\) does not change the norm, \(\norm{\tau^\Psi_t (A)} = \norm{A}\) but increases the support, if \(A\in \alg_X\), then \(\tau^\Psi_t (A)\in \alg_{X_R}\).
Hence, after applying a LRB for \(\tauint_{t,s}\) we need to subtract \(2R\) from the distance.

\section{Conclusions}
\label{sect:conclusions}

We proved that the quantum many-body dynamics produced by mutually commuting long-range interactions is much more strongly constrained than for general long-range interactions.
The difference is fundamental at a physical level, as commutativity directly affects the speed of information propagation and the shape of the light cone.
This finding opens up a new structural divide between general long-range interactions and long-range commuting ones.

While commuting interactions are very special, they are of great current interest in physics in the context of quantum error correction~\cite{DBT2021,BK2022, HMKL2023,PJBP2024,BGKL2024}.
An illustrative example of a long-range toric stabilizer code is given in \cref{ex:lrtc}.
The enhanced LRBs place severe limits on the rate of entanglement generation~\cite{BHV2006} and quantum messaging~\cite{EW2017} for these quantum codes.

The results also imply that commuting long-range Hamiltonians with \(\alpha>1\) have anomalously slow dynamics and could therefore provide simple test beds for mathematical physics conjectures surrounding many-body localization~\cite{AABS2019}.
For example, they satisfy the main assumption in~\cite{TB2024} without disorder-averaging.
Finally, it would be interesting to probe the robustness of the methods by assuming only power-law decay of the commutators, in the spirit of~\cite{PHKM2010} and~\cite{HHKL2021}, which treated such a question for finite-range and short-range interactions, respectively.

\statement{Acknowledgments}

We thank Tomotaka Kuwahara for useful comments on the first version of the manuscript.
\par\nopagebreak
The research of M.\,L.\ is  supported by the Deutsche Forschungsgemeinschaft (DFG, German Research Foundation)
through TRR 352 (470903074) and by the European Union (ERC, MathQuantProp, project 101163620).\footnote{Views and opinions expressed are however those of the authors only and do not necessarily reflect those of the European Union or the European Research Council Executive Agency. Neither the European Union nor the granting authority can be held responsible for them.}
The research of T.\,W.\ is  supported by the Deutsche Forschungsgemeinschaft (DFG, German Research Foundation) through TRR 352 (470903074) and FOR 5413 (465199066).

\statement{Conflict of interest}
The authors have no conflicts to disclose.

\statement{Author contributions}
\begin{description}[
        font=\normalfont,
        nosep,
        align=right,
        left=0pt,
        widest={Marius Lemm:},
    ]
    \raggedright
    \item[Marius Lemm:]
        conceptualization (equal);
        investigation (equal);\linebreak[2]
        writing -- original draft (equal);
        writing -- review and editing (equal).
    \item[Tom Wessel:]
        conceptualization (equal);
        investigation (equal);\linebreak[2]
        writing -- original draft (equal);
        writing -- review and editing (equal).
\end{description}

\statement{Data availability}
Data sharing is not applicable to this article as no new data were created or analyzed in this study.

\AtNextBibliography{\footnotesize}
\printbibliography[heading=bibintoc]

\end{document}